\documentclass[format=acmsmall, review=true]{acmart}
\usepackage{acm-ec-21}
\setcitestyle{acmnumeric}
\usepackage[lined,boxed,ruled,norelsize,algo2e,linesnumbered,noend]{algorithm2e}
\usepackage{booktabs}
\SetAlFnt{\small}
\SetAlCapFnt{\small}
\SetAlCapNameFnt{\small}
\SetAlCapHSkip{0pt}
\IncMargin{-\parindent}

\usepackage[english]{babel}
\usepackage[toc,page]{appendix}
\usepackage{hyperref}
\usepackage{wrapfig,lipsum}	
\usepackage[labelformat=empty]{caption}
\usepackage{yfonts}	
\usepackage{amsmath}
\usepackage{amsthm}
\usepackage{scalerel}
\usepackage{verbatim}
\usepackage{thm-restate}
\SetKw{Continue}{continue}

\newtheorem{theorem}{Theorem}[section]
\newtheorem{claim}[theorem]{Claim}
\newtheorem{remark}[theorem]{Remark}

\numberwithin{equation}{section}

\usepackage{times}

\usepackage{slivkins-setup,nicefrac}

\newcommand{\term}[1]{\text{\tt{#1}}\xspace}

\newcommand{\ALG}{\term{ALG}}
\newcommand{\TS}{\term{TS}}
\newcommand{\MainALG}{\term{ExponentialExploration}}

\newcommand{\PoI}{\term{PoIE}} 

\newcommand{\ZEROS}{\term{ZEROS}}
\newcommand{\ExplorePhase}[1][j]{\textsc{Exploration Phase} for arm $#1$}
\newcommand{\ExploitPhase}[1][N]{\textsc{Exploitation Phase} with depth $#1$\xspace}
\newcommand{\PaddedPhase}{\textsc{Padded Phase}\xspace}
\newcommand{\myTAB}{\hspace{2em}}

\newcommand{\padding}{\lambda}
\newcommand{\padG}{G_{\term{pad}}}
\newcommand{\padN}{N_{\term{pad}}}
\newcommand{\padC}{C_{\term{pad}}}

\newcommand{\bootN}{N_{\term{boot}}}
\newcommand{\bootP}{p_{\term{boot}}}

\newcommand{\RndsUB}{T_{\term{UB}}}     
\newcommand{\RndsUBtwo}{T_{\term{UB2}}} 
\newcommand{\RndsLB}{T_{\term{LB}}}   
\newcommand{\RndsOPT}{T_{\term{OPT}}} 

\newcommand{\MainLB}{L_{\term{LB}}} 

\newcommand{\post}{\term{pad}}
\newcommand{\Nexplore}{\bootN}
\newcommand{\Pexplore}{\bootP}

\newcommand{\Tbic}[1]{\RndsUB(#1)}

\usepackage{float}

\renewcommand{\E}{\mathbb{E}} 
\newcommand{\Var}{\mathtt{Var}} 

\newcommand{\brac}[1]{\sbr{#1}}
\renewcommand{\refeq}[1]{Eq.~(\ref{#1})}

\newcommand{\bigXhdr}[1]{\vspace{3mm}\noindent\textsc{\textbf{#1}}}

\usepackage[suppress]{color-edits}  
\addauthor{as}{red} 
\addauthor{ms}{blue}

\title{The Price of Incentivizing Exploration: \\
{\Large A Characterization via Thompson Sampling and Sample Complexity}}

\author{Submission 129}

\date{}

\begin{document}

\begin{abstract}
We consider \emph{incentivized exploration}: a version of multi-armed bandits where the choice of arms is controlled by self-interested agents, and the algorithm can only issue recommendations. The algorithm controls the flow of information, and the information asymmetry can incentivize the agents to explore. Prior work achieves  optimal regret rates up to multiplicative factors that become arbitrarily large depending on the Bayesian priors, and scale exponentially in the number of arms. A more basic problem of sampling each arm once runs into similar factors.


We focus on the \emph{price of incentives}: the loss in performance, broadly construed, incurred for the sake of incentive-compatibility. We prove that Thompson Sampling, a standard bandit algorithm, is incentive-compatible if initialized with sufficiently many data points. The performance loss due to incentives is therefore limited to the initial rounds when these data points are collected.
The problem is largely reduced to that of sample complexity: how many rounds are needed? We address this question, providing matching upper and lower bounds  and instantiating them in various corollaries. Typically, the optimal sample complexity is polynomial in the number of arms and exponential in the ``strength of beliefs".

\end{abstract}

\maketitle

\addtocontents{toc}{\protect\setcounter{tocdepth}{0}}

\newpage
\section{Introduction}

Consider an online platform where users need to choose between some actions (\eg products or experiences) of initially unknown quality, and can jointly learn which actions are better. The users collectively face the tradeoff between \emph{exploring} various actions so as to acquire new information, and \emph{exploiting} this information to choose better actions. A benevolent dictator controlling the users would run an algorithm to resolve this tradeoff so as to maximize social welfare. The online platform may wish to coordinate the users in a similar way. However, each user is a self-interested agent making her own choices, and her incentives are heavily skewed in favor of exploitation. This is because she suffers full costs of her exploration, whereas its benefits are spread among many agents in the future. Misaligned incentives can lead to under-exploration, whereby better alternatives are explored very slowly or not at all if they are unappealing initially.
These issues are common in online platforms that present recommendations and ratings based on user feedback, which are ubiquitous in a variety of domains: movies, restaurants, products, vacation destinations, etc. 

We study \emph{incentivized exploration}: the problem faced by the platform in the scenario described above. The platform can recommend actions, but cannot force agents to follow these recommendations. However, the platform controls the flow of information, and can choose what each agent observes about the past. Revealing full information to each agent works badly: agents fail to explore in a broad range of problem instances \citep[\eg see Ch. 11,][]{slivkins-MABbook}. Information asymmetry, when the platform reveals less than it knows, can incentivize the agents to explore.

A common model for incentivized exploration from \citep{Kremer-JPE14,ICexploration-ec15-conf} and the subsequent work is as follows. The population of agents faces a \emph{multi-armed bandit} problem, a basic model of exploration-exploitation tradeoff. A bandit algorithm iteratively recommends actions, a.k.a. \emph{arms}. In each round, a new agent arrives, observes a recommendation, chooses an action, and collects a reward for this action. This reward lies in the interval $[0,1]$, and comes from a fixed but unknown action-specific reward distribution. The reward is observed by the algorithm but not by the other agents. The algorithm does not reveal any information other than the recommended action itself; this is w.l.o.g. under standard economic assumptions of Bayesian rationality. In particular, the arms' mean rewards are drawn from a Bayesian prior which is known to everyone. The algorithm needs to be Bayesian incentive-compatible (\emph{BIC}), \ie incentivize the agents to follow recommendations. The goal is to design a BIC  bandit algorithm so as to optimize its learning performance.

Prior work on BIC bandit algorithms compares their learning performance to that of optimal bandit algorithms, BIC or not.
In particular, \citet{ICexploration-ec15-conf,ICexploration-ec15} obtain Bayesian regret
    $C_{K,\mP}\cdot \sqrt{T}$,
where $T$ is the time horizon and $C_{K,\mP}$ is determined by the number of arms $K$ and the Bayesian prior $\mP$; this dependence on $T$ is optimal in the worst case.
\footnote{Bayesian regret is a standard performance measure for Bayesian bandits (\ie multi-armed bandits with a Bayesian prior). It is defined as the difference in cumulative reward between the algorithm and the best arm, in expectation over the prior.}
However, $C_{K,\mP}$ can be arbitrarily large depending on the prior, even for $K=2$, and the dependence on $K$ is exponential in paradigmatic special cases. For contrast, non-BIC bandit algorithms achieve $O(\sqrt{KT})$ regret rate uniformly over all priors. Similar issues arise for a more basic variant of incentivized exploration, where one only needs to choose each arm at least once. This variant requires $C_{K,\mP}$ rounds in \citet{ICexploration-ec15-conf,ICexploration-ec15}, without any non-trivial lower bounds on the number of rounds, or any way to relate upper and lower bounds to one another.

\xhdr{Our scope.}
We focus on the \emph{price of incentivizing exploration} (\PoI): the penalty in performance incurred for the sake of the BIC property, such as the $C_{K,\mP}$ factor mentioned above. While several refinements of incentivized exploration have been studied, a more fundamental question of \textbf{characterizing the optimal \PoIbf} is largely open. This question is a unifying framing for our results.

While intuitive on a high level, the concept of \PoI is subtle to pin down formally. This is because the ``penalty in performance" can be expressed via different performance measures. Among these, we are particularly interested in Bayesian regret and \emph{sample complexity}: essentially, how many rounds are needed to choose each arm. Moreover, the increase in Bayesian regret could be multiplicative and/or additive, and is best measured relative to a particular bandit algorithm.%
\footnote{Formally, fix a near-optimal bandit algorithm $\mA^*$ and let $\E[R(\cdot)]$ denote Bayesian regret. Given a BIC algorithm $\mA$, write
    $\E[R(\mA)] = \alpha\cdot \E[R(\mA^*)] +\beta$.
Then $\alpha$, $\beta$ are, resp., multiplicative and additive increase in Bayesian regret.}

The question of characterizing the optimal \PoI comes in several flavors. First, what is the optimal dependence on $K$, the number of arms? For instance, when is this dependence polynomial as opposed to exponential? Second, what is the optimal dependence on the Bayesian prior? It is unclear what are the right parameters to summarize this dependence, and which properties of the prior make the problem difficult. In fact, it is not even clear if the dependence on the prior is needed.
Third, while BIC algorithms in prior work suffered from a multiplicative increase in Bayesian regret, it is desirable to make it additive.


We shed light on these issues, focusing on the canonical case of independent priors.
That is, the mean reward of each arm $i$ is drawn independently from the respective Bayesian prior $\mP_i$.


\xhdr{Our results: Thompson Sampling is BIC.}
We consider Thompson Sampling \citep{Thompson-1933}, a well-known bandit algorithm. We prove that Thompson Sampling is BIC given a \emph{warm-start}: a known number of samples of each arm, denoted $N_\TS$. More specifically, $N_\TS$ depends only on the Bayesian prior and the number of arms $K$, but not on the time horizon $T$ or the arms' mean rewards. Further, Thompson Sampling is BIC \emph{as is} when all arms have the same prior mean reward. Thompson Sampling is the first ``natural" bandit algorithm found to be BIC, whereas all BIC bandit algorithms from prior work are custom-designed.

This result has far-reaching implications. Thompson Sampling is widely recognized as a state-of-art algorithm for multi-armed bandits (or very close thereto), in terms of provable guarantees as well as empirical performance. In particular, it achieves the optimal $O(\sqrt{KT})$ Bayesian regret starting from any prior (or any warm-start). We view it as a ``gold standard" for bandits, as far as incentivized exploration is concerned. Therefore, the \PoI reduces to the performance loss due to collecting samples to warm-start Thompson Sampling. In particular, the increase in Bayesian regret is additive rather than multiplicative.

The warm-start size $N_\TS$ is an interesting measure of \PoI in its own right, as the initial samples  may be collected exogenously,
\eg purchased at a fixed price per sample.
We prove that $N_{\TS}$ is linear in $K$ under mild assumptions. Moreover, it can be as low as $O(\log K)$ for the natural example of Beta priors with bounded parameters. This is a huge improvement over \citep{ICexploration-ec15-conf,ICexploration-ec15}, where some (custom-designed) low-regret algorithms are proved to be BIC given some amount of initial data, but the necessary amount is not upper-bounded in terms of $K$ and can be (at least) exponential in $K$ for some examples. The $O(\log K)$ scaling is particularly appealing if each arm is contributed to the platform by a self-interested party, \eg it represents a restaurant that wishes to be advertised. Then each arm can be asked to pay an entry fee to subsidise the initial samples, and this fee only needs to scale as $O(\log K)$.


\revedit{Lastly, our analysis of Thompson Sampling implies an important \emph{monotonicity} property: its Bayesian-expected per-round reward is non-decreasing over time. Consequently, its Bayesian simple regret%
\footnote{Bayesian simple regret is another standard performance measure, defined as the difference in reward at round $t$ between the algorithm and the best arm, in expectation over the prior. }
 at each round $t$ is at most $O(\sqrt{K/t})$. These results appear new, and may be of independent interest.}

\xhdr{Our results: sample complexity.}
\revedit{
We turn to collecting initial samples of each arm, arguably the most basic variant of  incentivized exploration. The samples can be used to warm-start Thompson Sampling (or some other bandit algorithm with a similar BIC guarantee), and to estimate the expected rewards. More formally, we consider the following problem, called \emph{BIC $n$-sampling}: collect $n$ samples of each arm by a BIC bandit algorithm, in some number of rounds determined by the prior.%
\footnote{To appreciate why the number of rounds should be determined by the prior, consider a BIC algorithm that collects $n$ samples of each arm by some round $T_0$ that depends on the data. Suppose one runs this algorithm for $T_0$ rounds and then switches to Thompson Sampling. If $T_0$ depends on the data, then the combined algorithm is not necessarily BIC, as the timing of the switch could potentially leak information to the agents and alter their incentives.}
We are interested in minimizing this number of rounds; we call it the \emph{$n$-sample complexity}.



We provide ``polynomially matching" upper and lower bounds on the optimal $n$-sample complexity. In particular, we obtain \emph{the first non-trivial lower bound} specific to incentivized exploration, for any variant thereof (as opposed to lower bounds on regret from multi-armed bandits). The matching upper bound, \ie an algorithm and its analysis, is the most technical part of the paper. We use these bounds to resolve exponential vs. polynomial dependence on the number of arms ($K$) and the strength of beliefs, as expressed by one over the smallest variance $\sigma^2 = \min_i \Var(\mP_i)$. The common case is that the dependence on $K$ is polynomial, and the dependence on $1/\sigma$ is exponential. These are also \emph{upper bounds} for Bayesian regret of BIC $n$-sampling, which are new compared to prior work.%
\footnote{The exponential dependence on $1/\sigma$ was previously known as an upper bound, but only for $2$ arms \citep{ICexploration-ec15}.}
Thus, we characterize the additive \PoI of Thompson Sampling: it is polynomial in $K$ and exponential in $1/\sigma$ in terms of the number of rounds, and at most that much in terms of Bayesian Regret.

We emphasize that the $n$-sample complexity is an important performance measure on its own. This is because the platform may have various objectives instead of (or in addition to) Bayesian regret, and the $n$-sample complexity is meaningful for most/all of them.
\textbf{(i)} The platform may be interested in ``frequentist" performance guarantees (ones that hold for each realization of the prior), \eg as in \citet{ICexploration-ec15}. In particular, $n$-sample complexity upper-bounds frequentist regret of BIC $n$-sampling, and may plausibly be a good proxy for it.
\textbf{(ii)} The platform may be interested in ``pure exploration": predicting the best arm after a given number of rounds. In particular, the optimal $1$-sample complexity lower-bounds the number of rounds needed for any non-trivial frequentist guarantee on the prediction quality.
\textbf{(iii)} The platform's utilities for the actions may be different from the agents', \eg the former may be more forgiving for negative outcomes, and/or incorporate platform's revenue. Also, the platform may treat the prior as (merely) a \emph{belief} shared by the agents, and optimize in expectation over a different belief. In fact, the platform may wish to optimize with respect to multiple versions of utilities and/or beliefs. Yet, the $n$-sample complexity upper-bounds regret of BIC $n$-sampling with respect to any of them.
\textbf{(iv)} For two arms, the optimal $1$-sample complexity is the smallest number of rounds that guarantees \emph{any} non-trivial exploration almost surely.
} 

\xhdr{Sample complexity in more detail.}
We design a BIC bandit algorithm for collecting $n$ samples of each arm, called \MainALG. We prove that it runs for $\RndsUB(n)$ rounds for a given $n$, where $\RndsUB(n)$ is expressed in terms of the prior. We also provide a lower bound $\RndsLB$ on $1$-sample complexity: the number of rounds needed to choose each arm even once.%
\footnote{Our lower bound requires Bernoulli rewards. If arbitrary reward values are allowed, even a single random sample could reveal a huge amount of information to the algorithm. For example, its binary expansion could encode the mean reward.}
This lower bound is polynomially related to $\RndsUB(N_\TS)$ if each arm's prior has at least a constant variance: $\RndsUB(N_\TS) \leq \RndsLB^{O(1)}$; recall that $N_\TS$ is from Thompson Sampling. Thus, we characterize the optimal $n$-sample complexity for any $n\in [1,N_\TS]$, \ie both for warm-starting Thompson Sampling and for choosing each arm once. Moreover, $\RndsUB(N_\TS)$ upper-bounds optimal additive \PoI in terms of Bayesian regret.

We study how the optimal $n$-sample complexity, $n\in [1,N_\TS]$, depends on the number of arms $K$ and the smallest variance
    $\sigma^2 = \min_i \Var(\mP_i)$.
To isolate the dependence on $K$, we stipulate that the priors come from a fixed collection $\mC$, and study the worst-case dependence on $K$ over all such problem instances. We find a curious dichotomy: the dependence on $K$ is either always polynomial or can be exponential, depending on $\mC$.
An improved algorithm for the ``easy" case of this dichotomy achieves \emph{linear} dependence on $K$, which is  the best possible for a fixed $n$. Next, we focus on truncated Gaussian priors and Beta priors, two paradigmatic examples for Bayesian inference. We find that the optimal $n$-sample complexity is polynomial in $K$ and exponential in $\poly(1/\sigma)$. We conclude that the dependence on the priors cannot be avoided, and that \emph{strong beliefs}, as expressed by low-variance priors, is a key factor. 

Finally, we zoom in on the important special case when one arm $j$ represents a well-known alternative and all other arms are new:
    $\Var(\mP_j) \ll \Var(\mP_i)$ for all arms $i\neq j$.
Focusing on Beta priors, we prove that the exponential dependence on
    $\poly\rbr{1/\Var(\mP_j)}$
can be excluded from $n$-sample complexity. We replace it with a similar dependence on the second-smallest variance, and only a polynomial dependence on $1/\Var(\mP_j)$. This result is particularly clean for $K=2$ arms, whereby the sample complexity is driven by the larger variance
    $\max\rbr{ \Var(\mP_1),\,\Var(\mP_2)}$.
For $K>2$ arms, this result is the best possible: we prove that \emph{two} arms with small variance cannot be excluded in a similar fashion.


\xhdr{Explorability characterization.}
An important aspect of the \PoI is whether all arms are \emph{explorable}: can be sampled at least once by a BIC bandit algorithm. It is easy to construct examples when this is not the case. For instance, if there are two arms and $\E[\mu_2]<\mu_1$ almost surely, where $\mu_i$ is the mean reward of arm $i$, then arm $2$ cannot be explored.%
\footnote{In this case, Bayesian regret is $\Omega(T)$ provided that $\Pr[\mu_2-\mu_1>0]>0$.}
To ensure that all arms are explorable, we posit that
    $\E[\mu_i]>\mu_j$
with positive probability for all arms $i\neq j$. This condition, called \emph{pairwise non-dominance}, suffices for our results. Moreover, we prove that this condition is necessary for exploring all arms, under mild non-degeneracy assumptions. In fact, an arm $i$ is explorable if and only if it satisfies this condition, and all our results can be restricted to explorable arms. Thus, we provide a full characterization for which arms are explorable. This result complements several partial results from \citep{ICexploration-ec15-conf,ICexplorationGames-ec16,ICexploration-ec15}.%
\footnote{Specifically, a full characterization for $K=2$ arms, some sufficient conditions for $K>2$, and an algorithm that explores all ``explorable" arms but does not yield any explicit conditions. Unlike ours, these results extend to correlated priors.}

\xhdr{Our techniques.}
Algorithm \MainALG extends and amplifies the ``hidden exploration" approach from \cite{ICexploration-ec15-conf,ICexploration-ec15}, whereby one hides low-probability exploration amidst high-probability exploitation. We prove that exploration has a compounding effect: exploration in the present gives the algorithm more leverage to explore in the future, which allows the exploration probability to increase \emph{exponentially} over time. A simple algorithm design which branches out into ``pure exploration" and ``pure exploitation" in each round is no longer sufficient to realize these improvements. We introduce a new branch which combines exploration and exploitation so as to guarantee a stronger BIC property for one particular arm that is being explored. The latter property allows the algorithm to offset additional exploration of this arm (and does so more efficiently than ``pure exploitation").
The policy for this new branch is defined indirectly, as a maximin solution of a certain zero-sum game. The three branches are interleaved in a somewhat intricate way, to achieve a BIC algorithm with the above-mentioned exponential growth.

Our analysis of Thompson Sampling relies on martingale techniques, the FKG inequality (a correlation inequality from statistical mechanics), and a Bayesian version of Chernoff Bounds which appears new. When all prior means are the same, our analysis zooms in on the covariance between the posterior mean reward of one arm and the event that another arm is posterior-best.



\xhdr{Further discussion.}
We focus on a fundamental model of incentivized exploration that combines standard economic assumptions and a basic model of multi-armed bandits. Conceptually, this is the simplest model in which one can study the \PoI. Reality can be more complex in a variety of ways, both on the economics side and on the machine learning side (see related work for examples). However, our lower bounds immediately apply to models in which incentivized exploration is more difficult.

The standard economic assumptions mentioned above include common priors, agents' rationality and platform's commitment power. They are shared by all prior work on incentivized exploration (with a notable exception of \citet{Jieming-unbiased18}). Likewise, we assume that rewards are observed by the algorithm. Eliciting informative signals from the agents (\eg via reviews on an online platform such as Yelp or Amazon) is an important problem that is beyond our scope.

While we do not optimize absolute constants in the performance guarantees, several aspects of our results have practical appeal. We justify the usage of Thompson Sampling in incentivized exploration, reduce the problem to collecting initial data, and calibrate expectations for how much data is needed. An informal take-away is that Thompson Sampling with a moderately-sized batch of initial data should be BIC. Our results on $n$-sample complexity feature exponential improvements in the dependence on the number of arms and (for the scenario with ``one known arm") on the strength of beliefs.

\xhdr{Open questions.}
\revedit{The most immediate questions concern the dependence on the strength of agents' beliefs, as expressed by the smallest variance $\sigma^2 = \min_i \Var(\mP_i)$.
The first question is about the warm-start size $N_\TS$ for Thompson Sampling. Can it be made polynomial in $1/\sigma$? While it scales exponentially in $\poly(1/\sigma)$ in our result, we do not have any lower bounds.
The second question is about Bayesian regret for collecting $1$ sample of each arm. Can it be made polynomial in $1/\sigma$? We upper-bound it by  $1$-sample complexity, which in turn is lower-bounded by $\exp(\poly(1/\sigma))$. However, Bayesian regret of \MainALG is unclear. Our lower bound on $1$-sample complexity does not appear to have any bearing on Bayesian regret, so even a \emph{constant} dependence on $\sigma$ is not ruled out.
} 


Thompson Sampling as a technique applies far beyond the basic version of multi-armed bandits, and our results suggest it as a promising approach for more general models of incentivized exploration. One may hope to handle correlated priors and large-but-tractable bandit problems such as linear bandits. Likewise, one would like to extend our sample-complexity results to such problems.


Going back to independent priors, it is interesting whether other ``natural" bandit algorithms can be proved BIC given enough initial data. Our proof techniques are heavily tailored to Thompson Sampling. However, proving that such a result is impossible for a particular algorithm appears quite challenging, too.




\xhdr{Map of the paper.}
We analyze Thompson Sampling in Section~\ref{sec:TS}. Our sampling algorithm is presented and analyzed in Section~\ref{sec:sampling}. Various upper and lower bounds on sample complexity are spelled out in Section~\ref{sec:LB}. Section~\ref{sec:ext} contains some extensions. Various details are in appendices. In particular,
explorability characterization can be found in Appendix~\ref{sec:explorability}.

\section{Preliminaries}
\label{sec:prelim}

\textbf{Problem formulation: incentivized exploration.}
There are $T$ rounds and $K$ actions, a.k.a. \emph{arms}. In each round $t\in [T]$, an algorithm (a.k.a. the \emph{planner}) interacts with a new agent according to the following protocol. The algorithm recommends an arm $A_t$, the agent observes the recommendation (and nothing else) and chooses an arm $A'_t$ (not necessarily the same). The agent collects reward $r_t\in[0,1]$ for the chosen action, which is observed by the algorithm, but not by the other agents. The reward of each arm $i$ is drawn independently from some fixed distribution with mean $\mu_i\in [0,1]$. The reward distributions are initially not known to anybody. If agents always follow recommendations, \ie if $A'_t = A_t$ in all rounds $t$, then the problem reduces to (\revedit{Bayesian}, stochastic) multi-armed bandits.

We posit Bernoulli rewards, \ie $r_t \in \{0,1\}$ for all rounds $t$. This assumption is without loss of generality for all algorithmic results (\ie all results except the lower bounds). Essentially, this is because one can replace a reward $r\in [0,1]$ by a randomized Bernoulli reward with the same expectation.
\footnote{\revedit{The same trick applies to rewards of larger magnitude after re-scaling them to lie in $[0,1]$.}}

Let us specify Bayesian priors and incentives. For each arm $i$, the mean reward $\mu_i$ is independently drawn from prior $\mP_i$. (The joint prior is therefore
    $\mP_1 \times \ldots \times \mP_K$.)
The priors are known to all agents and the algorithm.  We require the algorithm to be \emph{Bayesian-incentive compatible} (\emph{BIC}): following recommendations is in the agents' best interest. Formally, we condition on the event that recommendations have been followed in the past,
    $\mE_{t-1} = \{A_s = A'_s:\; s\in [t-1]\}$.
The BIC condition is as follows: for all rounds $t$,
\begin{align}\label{eq:BIC-defn}
\E\sbr{ \mu_i - \mu_j \mid A_t=i,\; \mE_{t-1} }\geq 0
\quad\text{all arms $i,j$ such that $\Pr[A_t=i]>0$}.
\end{align}
If an algorithm is BIC, we assume that the agents actually follow recommendations.

As a stepstone towards BIC, we use a more restricted condition: a fixed subset $S$ of rounds is called BIC if \eqref{eq:BIC-defn} holds for all rounds $t\in S$. In this definition, we still require that recommendations at rounds $t\not\in S$ are followed, as per event $\mE_{t-1}$, even if these rounds are not necessarily BIC.

One could consider a more general version of the problem, in which the algorithm can reveal an arbitrary message $\sigma_t$ in each round $t$ and does not need to be BIC; then each agent $t$ chooses an arm $i$ which maximizes
    $\E[\mu_i \mid \sigma_t]$.
However, it is easy to show that restricting to BIC algorithms that (only) recommend arms is w.l.o.g. \citep{Kremer-JPE14}, by a suitable version of the \emph{revelation principle}. We also remark that agents in realistic situations are likely to not know exactly which round they arrive in. However, our BIC condition easily extends if agents instead have beliefs over their arrival times.

We posit a condition called \emph{pairwise non-dominance}: for each arm $i$,
\begin{align}\label{eq:pwise-cond}
\Pr\sbr{ \mu_j< \E[\mu_i] }>0
\quad \text{for all arms $j\neq i$}.
\end{align}
This condition is w.l.o.g.: essentially, each arm $i$ is explorable if and only if it satisfies \eqref{eq:pwise-cond}, see Section~\ref{sec:explorability}.


\xhdr{Recommendation policies.}
A \emph{recommendation policy} $\pi$ is a function that inputs a random signal $\mS$ and outputs an arm. More formally, let \emph{signal} $\mS$ be a random variable (taking values in some abstract set) that is jointly distributed with the mean rewards, in the sense that the tuple
     $(\mS; \mu_1 \LDOTS \mu_K)$
comes from some joint distribution. \revedit{A recommendation policy $\pi$ given signal $\mS$ is a mapping from $\text{support}(\mS)$ to arms.  In particular, one round of a bandit algorithm can be interpreted a recommendation policy, with signal $\mS$ being the algorithm's history up to this round. Likewise, if an algorithm invokes a recommendation policy, then (unless specified otherwise) the policy receives the algorithm's current history as a signal.

A natural version of the BIC property considers random variable $\pi(S)$ and posits that}
\begin{align}\label{eq:BIC-policy-defn}
\E\sbr{ \mu_i - \mu_j \mid \pi(\mS) =i }\geq 0
\quad\text{all arms $i,j$ such that $\Pr\sbr{ \pi(\mS)=i }>0$}.
\end{align}
If \eqref{eq:BIC-policy-defn} holds, we say that policy $\pi$ is \revedit{\emph{BIC given signal $\mS$}.}




\vspace{-1mm}
\xhdr{Conventions.}
We index arms by $i,j,k\in [K]$. We refer to them as ``arm $i$" or ``arm $a_i$" interchangeably.

Let $\mu_i^0 = \E[\mu_i]$ be the prior mean reward of each arm $i$. W.l.o.g., we order arms by their prior mean rewards:
    $\mu_1^0 \geq \mu_2^0 \geq \ldots \geq \mu_K^0$.
Let
    $A^* = \min\rbr{\argmax_j \mu_j}$
be the best arm, with this specific tie-breaking.
$\mF_t$ denotes the filtration generated by the chosen actions and the realized rewards up to (and not including) a given round $t$. We set
    $\E^t\sbr{\cdot} = \E\sbr{ \cdot \mid \mF_t}$
and
    $\Pr^t\sbr{\cdot} = \Pr\sbr{\cdot \mid \mF_t}$,
as a shorthand.
Sometimes we condition only on the first $N_j$ samples of each arm $j$, for some fixed $N_j$. Such a $\sigma$-algebra is called \emph{static} and denoted by $\mG_{(N_1,\dots,N_K)}$. In the special case when we condition on the first $N$ samples of arms $1 \LDOTS j$, the $\sigma$-algebra is denoted by $\mG_{N,j}$.

The set of all distributions over arms $1 \LDOTS k$ is denoted by $\Delta_k$.
If $q=(q_1,\dots,q_K)$ is a distribution over arms, the corresponding mean reward
is
    $\mu_q := \sum_i \mu_i q_i$.

\xhdr{Thompson Sampling.}
The core concept in Thompson Sampling is sampling from a Bayesian posterior. Given a random quantity $X$ determined by the mean rewards $(\mu_1 \LDOTS \mu_K)$, the Bayesian posterior at round $t$ is the conditional distribution of $X$ given $\mF_t$. A \emph{posterior sample} of $X$ at round $t$ (equivalently: given $\mF_t$) is an independent random draw from this distribution. Thompson Sampling is a very simple bandit algorithm: in each round $t$, the chosen arm $A_t$ is a posterior sample for the best arm $A^*$. In particular,
\begin{align}\label{eq:TS-def}
    \Pr\sbr{ A_t = i \mid \mF_t } = \Pr\sbr{ A^*=i \mid \mF_t }
    \quad \text{for each arm $i$}.
\end{align}
The algorithm is computationally efficient in various special cases, \eg for independent Beta priors and Bernoulli rewards, and for independent Gaussian priors (truncated or not) and Gaussian rewards.


\xhdr{Tools.}
We make use of FKG inequality, a correlation inequality which says that increasing functions of independent random variables are non-negatively correlated. We state it in
Appendix~\ref{app:mon-ineq}.

We also use a Bayesian concentration inequality. For a given arm $i$, it relates the posterior mean reward and an independent draw from the posterior distribution on $\mu_i$. We prove that both quantities are within $1/\sqrt{\text{\#samples}}$ of $\mu_i$. While reminiscent of Chernoff Bounds, which compare $\mu_i$ to the sample average, this result appears new. We state it below, and prove it in Appendix~\ref{app:Bayesian-concentration}.

\begin{restatable}[Bayesian Chernoff Bound]{lemma}{tschernoff}
\label{lem:TSchernoff}
Fix round $t$ and parameters $\eps,r>0$. Suppose $\mF_t$  almost surely contains at least $\eps^{-2}$ samples of a given arm $i$. Let $\hat\mu_i$ be a posterior sample for the mean reward $\mu_i$. Then for some universal absolute constant $C$ we have:
\begin{align}
\label{eq:BChernoff}
\Pr\sbr{ |\hat\mu_i-\mu_i|\geq r\eps }
&\leq
C\,e^{-r^2/C};
\\
\label{eq:BChernoff-2}
\Pr\sbr{ |\E[\mu_i|\mF_t]-\mu_i|\geq r\eps}
&\leq
C\,e^{-r^2/C}.
\end{align}
More generally, let $q=(q_1,\dots,q_K)$ be a distribution over arms. Let
    $\mu_q = \sum_i \mu_i\, q_i$ and $\hat\mu_q = \sum_i \hat\mu_i\, q_i$
be the corresponding mean reward and posterior sample.
Suppose $\mF_t$ almost surely contains at least $\eps^{-2}$ samples of each arm $i$ with $q_i\neq 0$.
Then \eqref{eq:BChernoff} holds with $i$ replaced with $q$.
\end{restatable}

\section{Incentivized Exploration via Thompson Sampling}
\label{sec:TS}

We prove that Thompson Sampling is BIC if initialized with enough samples of each arm.

\begin{theorem}\label{thm:TSsamples}
Let \ALG be a BIC bandit algorithm such that by some fixed time $T_0$ it almost surely collects at least
    $N_{\TS} = C_{\TS}\,\eps_{\TS}^{-2}\;\log \delta_{\TS}^{-1}$
samples from each arm, for a large enough absolute constant $C_{\TS}$, where
\begin{align}\label{eq:thm:TSsamples}
    \eps_{\TS} = \min_{i,j\in[K]}\; \E\sbr{ (\mu_i-\mu_j)_+}
    \quad\text{and}\quad
    \delta_{\TS} = \min_{i\in[K]}\;\Pr\sbr{A^*=i}.
\end{align}
Then running \ALG for $T_0$ rounds followed by Thompson sampling is BIC.
 \end{theorem}

\begin{remark}\label{rem:switch-time}
It is essential for the BIC property that the ``switching time" $T_0$ in Theorem~\ref{thm:TSsamples} is fixed in advance. In particular, switching to Thompson sampling as soon as \ALG collects enough samples could leak information and destroy the BIC property. For example, suppose \ALG has the following property: if arm $2$ is good, then \ALG w.h.p. spends a long time exploring this arm, and then does not play arm $2$ during some fixed time interval. Then if arm $2$ is recommended by Thompson sampling during the latter time interval, the agent will recognize that arm $2$ must be bad, and refuse to play it.
\end{remark}

\begin{remark}\label{rem:TS-strict}
\revedit{The BIC property in Theorem~\ref{thm:TSsamples} can be made to hold with a quantifiable margin: specifically, we can ensure that the right-hand side in \refeq{eq:BIC-defn} is $\eps_{\TS}/2$.}
\end{remark}

\begin{remark}\label{rem:TS-exogenous}
\revedit{Our analysis of Thompson Sampling is oblivious to where the warm-up data is coming from. In particular, the data can be collected by a non-BIC algorithm, and agents' participation may be secured via other means, \eg monetary payments. However, one needs to ensure that Bayesian update on the warm-up data does not depend on the algorithm used to collect it. One could achieve this by reporting the full history of the data-collection algorithm, or, \eg, only reporting the first $N_{\TS}$ samples of each arm.}
\end{remark}

Let us investigate how $N_\TS$, the sample count from Theorem~\ref{thm:TSsamples}, depends on $K$, the number of arms. We find that $N_{\TS}= O_\mC(K)$ if all priors belong to a finite family $\mC$, and $N_{\TS}= O_\mC(\log K)$ if $\mC$ consists of all Beta priors of bounded strength of beliefs.

\begin{corollary}\label{cor:TSsamples}
Suppose all priors $\mP_i$, $i\in [K]$ come from some fixed, finite collection $\mC$ of priors which satisfy the pairwise non-dominance condition \eqref{eq:pwise-cond}. Then $N_{\TS}= O_\mC(K)$.
\end{corollary}

\begin{proof}
In fact we have $N_{\TS} = O(K\;\eps_{\TS}^{-2}\;\log \eps_\TS^{-1})$.
Indeed for each arm $i$,
    \[\Pr[A^*=i]\geq \prod_j\Pr[\mu_i\geq \mu_j]\geq \eps_{\TS}^K.\]
The latter inequality holds simply because $\Pr[\mu_i\geq \mu_j]\geq \mathbb E[(\mu_i-\mu_j)_+]$ for all $i,j$. Let $\eps_{\mC}$ be the version of $\eps_\TS$ where the $\min$ is over all ordered pairs of (not necessary distinct) priors in $\mC$. Note that $\eps_\TS\geq \eps_\mC$. Since $\mC$ is finite and satisfies pairwise non-dominance, $\eps_{\mC}$ is strictly positive. It remains to show the first part of the inequality above.

We proceed via the FKG inequality. Define the indicator functions $I_j=1_{\mu_i\geq \mu_j}$ for each $j\in [K]$; we interpret them as functions of $(\mu_1 \LDOTS \mu_K)$. The functions $I_j$ are each increasing in $\mu_i$ and decreasing in $\mu_{\ell}$ for all ${\ell}\neq i$ \revedit{(including $\ell=j$ when $j\neq i$)}. As the values $\{\mu_{\ell}\}_{{\ell}\in [K]}$ are independent by assumption, the mixed-monotonicity FKG inequality (see Remark~\ref{rem:fkgmixed}) implies that the indicator functions $I_j$ are non-negatively correlated. In fact, each product $\prod_{j=1}^{\ell}I_j$ satisfies the same monotonicity properties, so repeated application of mixed-monotonicity FKG implies
\begin{align*}
\Pr[A^*=i]
    &= \textstyle \E\left[\prod_{j=1}^K I_j\right]\\
    &\geq \textstyle \E\left[\prod_{j=1}^{K-1}I_j\right]\cdot \E[I_K]
    \geq \E\left[\prod_{j=1}^{K-2}I_j\right]\cdot \E[I_{K-1}]\cdot\E[I_K]
    \geq \dots\\
    &\geq \prod_{j=1}^K \E[I_j]
    =\prod_j \Pr[\mu_i\geq \mu_j].\qedhere
\end{align*}
\end{proof}

\begin{corollary}\label{cor:tslogk}
Suppose each prior $\mP_i$, $i\in [K]$ is a $\term{Beta}(\alpha_i,\beta_i)$ distribution with parameters $\alpha_i,\beta_i\in [1,M]$, for some fixed $M$. Then
    $N_{\TS}= O_M(\log K)$.
\end{corollary}

\begin{proof}
Note that
    $\eps_{\TS}=\Omega(9^{-M})$.
This immediately follows from definition of $\eps_{\TS}$ in \refeq{eq:thm:TSsamples}, because
    $\Pr[\mu_i>\nicefrac{2}{3}] \geq 3^{-M}$
and
    $\Pr[\mu_j<\nicefrac{1}{3}] \geq 3^{-M}$
for all arms $i,j$.

To handle $\delta_\TS$, let $Q_p[\mP_i]$ be the top $p$-th quantile of distribution $\mP_i$. Suppose for some $\eta>0$
\begin{align}\label{eq:rem:TSsamples}
  Q_{\eta}[\mP_i] \geq Q_{1/K}[\mP_j] \quad\text{for all arms $i,j$.}
\end{align}
Then
$\delta_{\TS}
    \geq \Pr[A^*=a_i]
    \geq \eta\, \rbr{1-\nicefrac{1}{K}}^{K-1}
    \geq \Omega(\eta)$
for all arms $i$, so $N_{\TS} = O_M(\log \nicefrac{1}{\eta})$.

To complete the proof, we claim that \eqref{eq:rem:TSsamples} holds with
    $\eta = (MK)^{-M}$.
This is because for each arm $i$ we have
    $(MK)^{-M} \leq \Pr\sbr{ \mu_i>1-\tfrac{1}{MK}}\leq \tfrac{1}{K}$.
To verify the last statement, it suffices to focus on the extremal cases $\term{Beta}(1,M)$ and $\term{Beta}(M,1)$.
\end{proof}

Moreover, we prove that Thompson sampling is BIC \emph{as is} if all prior mean rewards are the same.

\begin{theorem}\label{thm:TS-BIC-simple}
If $\mu_1^0=\mu_2^0=\dots=\mu_K^0$ then Thompson sampling is BIC.
\end{theorem}



\subsection{Proofs}

We prove Theorems~\ref{thm:TSsamples} and~\ref{thm:TS-BIC-simple}.
First, we note that for any algorithm and any arms $i,j$ it holds that
\begin{align}
 \E[\mu_i \mid A_t=j]\cdot \Pr[A_t=j]
    &=\E\brac{ \mu_i\cdot \ind{A_t=j} } \nonumber\\
    &=\E\brac{ \E^t[\mu_i\cdot \ind{A_t=j}] }
    =\E\brac{ \E^t[\mu_i]\cdot \E^t[\ind{A_t=j}] } \nonumber\\
    &= \E\brac{ \E^t[\mu_i]\cdot {\Pr}^t[A_t=j] }. \label{eq:TS-bayes} \\
\E[\mu_i -\mu_j \mid A_t=j]\cdot \Pr[A_t=j]
    &= \E\brac{ \E^t[\mu_i-\mu_j]\cdot {\Pr}^t[A_t=j] }. \label{eq:TS-bayes-2}
\end{align}
(\refeq{eq:TS-bayes-2} follows by taking a version \eqref{eq:TS-bayes} with $i=j$, and subtracting it from \eqref{eq:TS-bayes}.)

Next, we analyze the object inside the expectation in \eqref{eq:TS-bayes}.

\begin{lemma}\label{lem:submartingale}
Fix arms $i,j$. Let
    $H_t:=\E^t[\mu_i]\cdot \Pr^t[A^*=j]$.
For any algorithm, the sequence
    $(H_1 \LDOTS H_T)$
is a supermartingale if $i\neq j$ and a submartingale if $i=j$.
\end{lemma}

\begin{proof}
Note that
    $(\E^t[\mu_i]:\; t\in [T])$ and $(\Pr^t[A^*=j]:\; t\in [T])$
are martingales by definition.

We consider two cases, depending whether $A_t=i$. \revedit{Recall that an expression such as $\E^{t+1}[\mu_i]$ is a random variable (with randomness coming from $\mF_{t+1}$), and event such as $A_t=i$ restricts this random variable.}
First, suppose $A_t\neq i$. Then $\E^{t+1}[\mu_i]=\E^t[\mu_i]$ almost surely, therefore $H_t$ has expected change $0$ on this step since $\Pr^t[A^*=j]$ is a martingale.

Next, suppose $A_t=i$. The crucial claim is that $\E^{t+1}[\mu_i]$ and $\Pr^{t+1}[A^*=i]$ are increasing in the time-$t$ reward while $\Pr^{t+1}[A^*=j]$ is decreasing in the time-$t$ reward. Indeed, Corollary~\ref{cor:1sampledom} and Lemma~\ref{lem:mlr} in the Appendix imply that the conditional distribution of $\mu_i$ is stochastically increasing in the time-$t$ reward. Observing that the event $\{A^*=j\}$ is decreasing in the value of $\mu_i$ now implies the claim. Note that this argument crucially uses both the Bernoulli reward assumption and the fact that we have bandit feedback and independent arms.

Next, recall the FKG inequality (see Appendix~\ref{app:tools}): if $f,g$ are increasing functions of the same variable then they are positively correlated, i.e. $\E[fg]\geq\E[f]\E[g]$. We set $f=\E^{t+1}[\mu_i]$ and $g=\Pr^{t+1}[A^*=i]$ and apply FKG conditionally on $\mF_t$ and $A_t$, interpreting both $f$ and $g$ as functions of the time-$t$ reward. Then
\begin{align*}
\E^t[H_{t+1}]&=\E^t\left[\E^{t+1}[\mu_i]\cdot {\Pr}^{t+1}[A^*=i]\right]\\
&\geq \E^{t+1}[\mu_i]\cdot {\Pr}^{t+1}[A^*=i]\\
&=H_t.
\end{align*}
We have just shown that $H_t$ is a submartingale when $i=j$. Similarly, when $i\neq j$ we apply the FKG inequality to $f=\E^{t+1}[\mu_i]$ and $g=\Pr^{t+1}[A^*=j]$, again conditionally on $\mF_t$ and $A_t$. In this case $g$ is a decreasing function of the time-$t$ reward and so the FKG inequality goes in the opposite direction, stating that $\E[fg]\leq\E[f]\E[g]$. We hence obtain
\begin{align*}
\E^t[H_{t+1}]&=\E^t\left[\E^{t+1}[\mu_i]\cdot {\Pr}^{t+1}[A^*=j]\right]\\
&\leq\E^{t+1}[\mu_i]\cdot {\Pr}^{t+1}[A^*=j]\\
&=H_t. \qedhere
\end{align*}
\end{proof}

The following lemma is essentially an inductive step. It implies Theorem~\ref{thm:TS-BIC-simple} by induction on $t$, because the premise in the lemma holds trivially when $t=0$ and all prior mean rewards are the same.

\begin{lemma}
\label{thm:TSBIC}
Let \ALG be any bandit algorithm. Fix round $t$. Suppose that running \ALG for $t-1$ steps, followed by Thompson sampling at time $t$, is BIC at time $t$. Then running \ALG for $t$ steps, followed by Thompson sampling at time $t+1$, is BIC at time $t+1$.
\end{lemma}

\begin{proof}
Thompson sampling is BIC at time $t$ if and only if
    $\E[\mu_i-\mu_j \mid A_t=i]\geq  0$ for all $(i,j)$.
\begin{align*}
\E[\mu_i-\mu_j \mid A_t=i]
   =&\frac{\E\sbr{ \E^t[\mu_i-\mu_j]\cdot \Pr^t[A_t=i]}}{\Pr[A_t=i]}
    &\EqComment{by \refeq{eq:TS-bayes-2}}\\
   =& \frac{\E\sbr{ \E^t[\mu_i-\mu_j]\cdot \Pr^t[A^*=i]}}{\Pr[A^*=i]}
    &\EqComment{by \refeq{eq:TS-def}}.
\end{align*}
In the numerator,
    $\E^t[\mu_i-\mu_j]\cdot \Pr^t[A^*=i]$
is a submartingale by Lemma~\ref{lem:submartingale}. In particular its expectation
is non-decreasing in $t$. On the other hand the denominator $\Pr[A^*=i]$ is a constant independent of $t$.
\end{proof}

\begin{proof}[Proof of Theorem~\ref{thm:TSsamples}]
Fix arms $i,j$ and set $\delta_i:=\Pr[A^*=i]\geq \delta_{\TS}$. Then
\begin{equation}
\label{eq:TSsamples-calc}
\begin{aligned}
 \E[\mu_i -\mu_j \mid A_{T_0}=i]\cdot \Pr[A_{T_0}=i]
 =&
 \E\brac{ \E^{T_0}[\mu_i-\mu_j]\cdot {\Pr}^{T_0}[A_{T_0}=i] }
 &\EqComment{by \refeq{eq:TS-bayes-2}}\\
 =&\E\brac{ \E^{T_0}[\mu_i-\mu_j]\cdot {\Pr}^{T_0}[A^*=i] }
  \\
 =&\E\brac{ \E^{T_0}[\E^{T_0}[\mu_i-\mu_j]\cdot \ind{A^*=i}] }
 \\
 =&\E\brac{ \E^{T_0}[\mu_i-\mu_j]\cdot \ind{A^*=i} }.
\end{aligned}
\end{equation}
To establish that Thompson Sampling is BIC we prove
    $\E\brac{ \E^{T_0}[\mu_i-\mu_j]\cdot \ind{A^*=i} } \geq 0$. Since the functions $(\mu_i-\mu_j)_+,\ind{A^*=i}$ are increasing in $\mu_i$ and decreasing in $\mu_k$ for each $k\neq i$, the FKG inequality implies
\[
    \E[(\mu_i-\mu_j)\cdot \ind{A^*=i}]=\E[(\mu_i-\mu_j)_+\cdot \ind{A^*=i}]\geq \eps_{\TS}\,\delta_{i},
\]
    see Remark~\ref{rem:fkgmixed}. If our estimates $\E^{T_0}[\mu_i],\E^{T_0}[\mu_j]$ of $\mu_i,\mu_j$ were exactly correct then we could immediately conclude. Inspired by this, we show the expected absolute error in estimating
    $ \E^{T_0}[\mu_i-\mu_j]\cdot \ind{A^*=i}$
by $(\mu_i-\mu_j)\cdot \ind{A^*=i}$ is upper bounded by $\eps_{\TS}\,\delta_{i}$. In other words we aim to show:
\begin{equation}\label{eq:TSapprox}
    \E\left[\left|\E^{T_0}[\mu_i-\mu_j]\cdot \ind{A^*=i} -(\mu_i-\mu_j)\cdot \ind{A^*=i}\right|\right]\leq \eps_{\TS}\,\delta_{i}.
\end{equation}
By the triangle inequality, establishing Equation~\eqref{eq:TSapprox} will complete the proof. By regrouping and using again the triangle inequality as well as $|x\cdot \ind{A^*=i}|=|x|\cdot \ind{A^*=i}$ for any $x\in\mathbb R$, the left-hand side of~\eqref{eq:TSapprox} is upper bounded by
\begin{align}\label{eq:pf-TS-1}
  \E\brac{ \left|  \E^{T_0}[\mu_i]-\mu_i \right| \cdot \ind{A^*=i} } +
    \E\brac{ \left|\E^{T_0}[\mu_j]-\mu_j \right| \cdot \ind{A^*=i} }.
\end{align}

By Lemma~\ref{lem:TSchernoff}, the values $\left|  \E^{T_0}[\mu_i]-\mu_i \right|$ and $\left|\E^{T_0}[\mu_j]-\mu_j \right|$ are $N_{\TS}^{-1/2}$ times $O(1)$-sub-Gaussian random variables. Applying Lemma~\ref{lem:subgaussiantail}, we obtain that both terms in \refeq{eq:pf-TS-1} are at most $O\left(\delta_{i}\sqrt{\log(1/\delta_{i})/N_{\TS}}\right)$. Using $\delta_{\TS}\leq \delta_i$ and our choice of $N_\TS$ we now conclude.
\end{proof}

\begin{proof}[Proof of Remark~\ref{rem:TS-strict}]
\revedit{Increasing the value of $N_{\TS}$ by a factor $4$, compared to Theorem~\ref{thm:TSsamples}, ensures that
\begin{equation}\label{eq:TSapprox-2}
    \E\sbr{ \left|\,\E^{T_0}[\mu_i-\mu_j]\cdot \ind{A^*=i} -(\mu_i-\mu_j)\cdot \ind{A^*=i}\,\right| }
    \leq \eps_{\TS}\,\delta_{i}/2
\end{equation}
Revisiting the above proof, this directly implies
\[
    \E\brac{ \E^{T_0}[\mu_i-\mu_j]\cdot \ind{A^*=i} }
    \geq
    \eps_{\TS}\,\delta_{i}/2.
\]
Recall \eqref{eq:TSsamples-calc} and the fact that $\delta_i=\Pr[A^*=i]=\Pr[A_{T_0}=i]$ (the latter equality holds by the definition of Thompson sampling and the martingale property of $\Pr^t[A^*=i]$). Combining implies
\[
      \E\sbr{\mu_i -\mu_j \mid A_{T_0}=i}
        \geq \eps_{\TS}/2. \qedhere
\]
} 
\end{proof}


\subsection{Monotonicity of Thompson Sampling}
\label{sec:TS-extensions}

Our analysis of Thompson Sampling implies that its expected reward grows monotonically with time, which in turn allows us to upper-bound its Bayesian simple regret
    $\E_{\text{prior}}\sbr{ \max_{i\in[K]} \mu_i  - \mu_{A_t}}$.
These results are new in the literature on Thompson Sampling, to the best of our knowledge.

\begin{corollary}
For Thompson Sampling (starting with an arbitrary prior), $\E_{\text{prior}}\sbr{\mu_{A_t}}$ is non-decreasing in round $t$. Consequently, Bayesian simple regret at each round $t$ is at most $O(\sqrt{K/t})$.
\end{corollary}

\begin{proof}
Denote
    $H_{t,i} = \E^t\sbr{\mu_i}\, \Pr^t[A_t=i]$,
a key object in the proof of Theorem~\ref{thm:TSsamples}. Then
\begin{align*}
\E\sbr{\mu_{A_{t+1}}}
    &= \sum_{\text{arms $i$}} \E\sbr{H_{t+1,i}}
    \geq \sum_{\text{arms $i$}} \E\sbr{H_{t,i}}
    &\EqComment{since $(H_{t,i}:\, t\in\N)$ is a submartingale} \\
    &= \sum_{\text{arms $i$}} \E[\mu_i \mid A_t=i]\; \Pr[A_t=i]
    &\EqComment{by \eqref{eq:TS-bayes-2}}\\
    &= \E\sbr{\mu_{A_t}}.
\end{align*}
The bound on Bayesian simple regret follows simply because the (cumulative) Bayesian regret at round $t$ is upper-bounded as $O(\sqrt{Kt})$ by \cite[Theorem 1]{bubeck2013prior}, and equals the sum of Bayesian simple regret over all rounds $s\in[t]$.
\end{proof}

\OMIT{ 

Consider an extension of incentivized exploration in which each agent privately observes $n$ independent samples of each arm, for some fixed $n$. We prove that Thompson Sampling is BIC with respect to this signal, provided that the warm-up data is increased by the factor of $\exp(O(n))$.

\begin{theorem}\label{thm:TSsamples-private}
Fix $n$. For each round $t$, let $\mS_{t,n}$ be a tuple of $n$ independent samples of each arm. Let \ALG be a BIC bandit algorithm such that by some fixed time $T_0$ it almost surely collects at least
    $N_{\TS}\cdot c^n$
samples from each arm, where $N_{\TS}$ is from Theorem~\ref{thm:TSsamples} and $c$ is a large enough absolute constant. Then running \ALG for $T_0$ rounds followed by Thompson sampling satisfies the following for each round $t>T_0$:
\begin{align}\label{eq:BIC-defn-private}
\E\sbr{ \mu_i - \mu_j \mid A_t=i,\; \mS_{t,n},\;\mE_{t-1} }\geq 0
\quad\text{all arms $i,j$ such that $\Pr[A_t=i]>0$}.
\end{align}
\end{theorem}

} 

\section{Collecting Initial Samples}
\label{sec:sampling}

\newcommand{\IndPadded}{\indE{\term{padded}}}
\newcommand{\IndExplore}{\indE{\term{explore}}}
\newcommand{\IndExploit}{\indE{\term{exploit}}}
\newcommand{\IndFlag}{\indE{\term{exploreflag}}}
\newcommand{\IndZero}{\indE{\ZEROS_{j,N_0}}}
\newcommand{\padExpr}{\Lambda^{\term{pad}}_{i,j}}

We turn to a basic version of incentivized exploration: collect $N$ samples of each arm. We design a BIC algorithm, called \MainALG, which completes after a pre-determined number of rounds and collects $N$ samples of each arm almost surely.
(This is a desirable property as per Remark~\ref{rem:switch-time}.)
We bound the completion time in terms some parameters of the prior.




We describe the algorithm on a high level, and then fill in the details. The algorithm explores the arms in order of increasing index $j$, \ie in the order of decreasing prior mean reward. A given arm $j$ is explored as follows. We partition time in phases of $N$ rounds each, where $N$ is a parameter. Within a given phase, the algorithm recommends the same arm in all rounds. It uses phases of three types: \emph{exploration phases}, when arm $j$ is always recommended, \emph{exploitation phases}, when the algorithm chooses an arm with the largest posterior mean reward, and \emph{padded phases}, which combine exploration and exploitation. In the exploitation phase, the algorithm conditions on the first $N'$ samples of each arm $i<j$, where $N'$ is a given \emph{depth} parameter. The algorithm also conditions on the first $N'$ samples of arm $j$ if these samples are available before the phase starts. In the padded phase, it leverages the samples from arm $j$  
to guarantee a stronger BIC property for this arm which is ``padded" by some prior-dependent amount $\padding >0$;
this property offsets more exploration for arm $j$. A given phase is assigned to one of these three types in a randomized and somewhat intricate way.

\newcommand{\FakePi}[1][j]{\term{transform}(\pi_{#1})} 
\newcommand{\happened}{exploration phase has happened\xspace}
\newcommand{\flag}{\term{flag}}
\newcommand{\invariant}[2]{\textcolor{blue}{\tcp*[h]{Invariant #1:} #2.}} 

\SetKwFor{Loop}{loop}{}{end}

\begin{algorithm2e}[t]
\caption{\MainALG}
\label{alg:pyramidtight}
\SetAlgoLined\DontPrintSemicolon

\textbf{Parameters:}
    phase length $N$,
    bootstrapping parameter $N_0\leq N$,
    padding $\padding>0$.

\textbf{Given:}
recommendation policies $\pi_2 \LDOTS \pi_K$ for \PaddedPhase.

\textbf{Initialize:} \ExplorePhase[1]

\For{each arm $j=2, 3  \LDOTS K$}{ \label{line:algloop}
    \invariant{1}{each arm $i<j$ has been sampled at least $N$ times}

    \vspace{1mm}\tcp*[l]{Bootstrapping: two phases}
    Event $\ZEROS_{j,N_0} =
        \cbr{ \text{the first $N_0$ samples of each arm $i<j$ return reward $0$}} $. \\
    $p_j\leftarrow q/(1+q)$, where $q = \padding\cdot\Pr\sbr{\ZEROS_{j,N_0}}$.
    \label{line:init-prob}\\ 
    \ExploitPhase[N_0] \label{line:initialexploit}

    \textbf{with probability} $p_j$ \textbf{do}\\
        \myTAB\ExplorePhase \label{line:alg1explorephase} \\
    \uElseIf{$\ZEROS_{j,N_0}$\label{line:zeros}}
        {\PaddedPhase: use policy $\FakePi$\label{line:padded1}}
                \lElse{\ExploitPhase[N]}

    \vspace{2mm}\tcp*[l]{main loop: exponentially grow the exploration probability}

    \While{$p_j< 1$}{
    \invariant{2}{$\Pr\sbr{\text{\happened} \mid \mu_1 \LDOTS \mu_K}=p_j$}

    \uIf{\happened}{
        \PaddedPhase: use policy $\pi_j$ \label{line:padded2}
        }
    \textbf{else with probability} $\min\rbr{1,\,\frac{p_j}{1-p_j}\cdot \padding}$ \textbf{do}\\
    \myTAB \ExplorePhase \label{line:algexplore} \\
    \lElse{\ExploitPhase}
Update $p_j\leftarrow \min\rbr{1,\, p_j\, (1+\padding)}$.
}
} 
\end{algorithm2e}

Algorithm~\ref{alg:pyramidtight} presents the algorithm with abstract parameters $N_0\leq N$ and $\padding>0$, and recommendation policies for the padded phase. Let us again focus on exploring a particular arm $j$. The first two phases (\emph{bootstrapping}) ensure that exploration phase is invoked with probability $p_j$ for any given vector of mean rewards $\mu_1 \LDOTS \mu_K$, where $p_j$ is given in Line~\ref{line:init-prob}. We capture this condition as Invariant 2. Then the algorithm enters the main loop, where it exponentially grows the exploration probability. More precisely, consider the \emph{phase-exploration probability}: the probability that the ``pure exploration" phase for arm $j$ has been invoked. The algorithm increases this probability by a $1+\padding$ factor after each iteration, maintaining Invariant 2. This exponential growth is the key aspect of the algorithm, which side-steps the fact that the initial phase-exploration probability in Line~\ref{line:init-prob} may be very small. Inside the main loop, the algorithm randomizes between ``pure exploitation" and ``pure exploitation" with predetermined probability (chosen so as to guarantee that the $1+\padding$ increase in phase-exploration probability overall), and permanently switch to the padded phase once exploration phase has been invoked. The padded phase offsets the additional exploration in the same iteration. This process continues until the phase-exploration probability reaches $1$.

The bootstrapping phases hide a considerable amount of complexity which may be skipped at a first reading. Conceptually, we would like to implement the ``hidden exploration" approach from \citet{ICexploration-ec15-conf,ICexploration-ec15}, which randomizes between exploration and exploitation phases with some predetermined probability. While this approach may suffice for some ``well-behaved" priors, it appears to be insufficient more generally, in the sense that the phase-exploration probability depends on some additional prior-dependent parameters that are difficult to deal with. Instead, we combine exploration and exploitation in a more sophisticated way, as explained later in the section.

Let's make some observations that are immediate from the algorithm's specification.

\begin{claim}\label{cl:alg-obs}
Algorithm~\ref{alg:pyramidtight} samples each arm at least $N$ times with probability $1$, and completes in
    $(2K-1)N+N\,\sum_{j=2}^K \Cel{\,\ln_{1+\padding} (p_j)\,}$
rounds, where $p_j$ is given in Line~\ref{line:init-prob}. Both invariants hold at each iteration of the respective loops. Bootstrapping takes exactly two phases, and each iteration of the \term{while} loop takes exactly one phase.
\end{claim}

\revedit{ 
\xhdr{Recommendation policies} $\pi_j$ need to satisfy several properties. Fix arm $j$. First, we require $\pi_j$ to input only the first $N$ samples of each arm $i\leq j$, ignoring the order in which the arms were sampled by the algorithm.%
\footnote{Formally, $\pi_j$ inputs an ordered tuple of arm-reward pairs, and  pre-processes it as a $j\times N$ matrix whose $(i,n)$-th entry, $(i,n)\in [j]\times[N]$, is the reward from arm $i$ from the $n$-th time it was sampled. The policy is then determined by this matrix.}
Such policies are called $(j,N)$-informed. Note that the algorithm has sufficient data to compute policy $\pi_j$ thanks to Invariant 1.

Second, we require policy $\pi_j$ to be BIC. Formally, we let $\mS_{j,N}$ be a signal that consists of exactly $N$ independently realized samples of each arm $i\leq j$, and we require $\pi_j$ to be BIC w.r.t this signal.

Third, we require $\pi_j$ to satisfy a stronger BIC property for arm $j$:
\begin{align}\label{eq:padded-BIC-defn}
\E\sbr{ (\mu_j - \mu_i)\cdot \ind{\pi_j(\mS_{j,N})=j}}\geq \padding
\quad\text{for each arm $i<j$}.
\end{align}
If \eqref{eq:padded-BIC-defn} holds, we say that policy $\pi_j$ is \emph{$(j,\padding)$-padded BIC}, where $\lambda$ is the ``padding".
The left-hand side in \eqref{eq:padded-BIC-defn} is the expected loss for the $j\to i$ swap: the expected loss when one starts with policy $\pi_j$ and replaces arm $j$ with arm $i$ whenever arm $j$ is recommended. Note that we integrate over the event that arm $j$ being chosen, rather than condition on this event. We recover the ``usual" BIC property for arm $j$ when $\padding=0$.

If policy $\pi_j$ satisfies all three properties,
it is called \emph{$(j,\padding,N)$-suitable}.

While BIC and $(j,\padding)$-padded BIC properties of $\pi_j$ are defined relative to signal $\mS_{j,N}$, one could also define them relative to any other signal $\mS$ that almost surely contains at least $N$ samples of each arm $i\leq j$. It is easy to see that these definitions are equivalent: any $(j,N)$-informed policy is BIC relative to signal $\mS_{j,N}$ if and only if it is BIC relative to signal $\mS$; likewise for the $(j,\padding)$-padded BIC property. This point allows us to analyze $(j,\padding,N)$-suitable policies abstractly, regardless of where exactly their input comes from.

} 

\xhdr{Padding and the main loop.}
The key is to specify what happens in the ``padded phase", and argue about  incentives that it creates. We use the properties of policies $\pi_j$, as listed above, to guarantee that the main loop is BIC. The main point is that the padded-BIC property compensates for the probability of new exploration. We carefully spell out which properties are needed where. In particular, the exploitation phase is only used to skip the round in a BIC way.%
\footnote{This is an interesting contrast with ``hidden exploration" \citep{ICexploration-ec15-conf,ICexploration-ec15}, where the exploitation phase is used to offset exploration. We have little use for this, because the padded phase enables nearly as much additional exploration as possible.}
In fact, the algorithm would work even if the exploitation phase in the main loop would always choose arm $1$.  However, using the available data for exploitation only improves the algorithm's efficiency as well as the incentives.

\begin{lemma}\label{lm:main-BIC}
Consider Algorithm~\ref{alg:pyramidtight} with arbitrary parameters $N_0\leq N$ and $\padding>0$.  Fix arm $j\geq 2$ and focus on the respective iteration of the \term{for} loop of the algorithm. Assume that policy $\pi_j$ is $(j,\padding,N)$-suitable. Then the \term{while} loop is well-defined and BIC.
\end{lemma}

\begin{proof}
The \term{while} loop is well-defined because policy $\pi_j$ is $(j,N)$-informed, so by Invariant 1 the algorithm has a sufficient amount of data to implement it.

Fix some round $t$ in the \term{while} loop. Let us restate the BIC property when arm $j$ is recommended:
\begin{align}\label{eq:BICcondition}
\E\sbr{ \Lambda_{i,j}} \geq 0,
\text{ where } \Lambda_{i,j}:=(\mu_j-\mu_i)\cdot \ind{A_t=j},
\end{align}
for all arms $i\neq j$. Let $\IndPadded$, $\IndExploit$, $\IndExplore$ be the indicator variables for the event that round $t$ is assigned to, resp., a padded, explotation, or exploration phase.  Due to Invariant 2, these indicator variables are independent of mean rewards $\mu_1 \LDOTS \mu_K$. We write
    \[\E\sbr{ \Lambda_{i,j}]=\E[\Lambda_{i,j}\cdot \rbr{\IndPadded+\IndExploit+\IndExplore}}\]
and estimate each expectation separately. Let
    $\padExpr :=(\mu_j-\mu_i)\cdot \ind{\pi_j=j}$
and observe that
\begin{align} \label{eq:lm:main-BIC:proof-padded}
\E\sbr{ \Lambda_{i,j}\cdot \IndPadded}
    =\E[\; \padExpr \cdot \IndPadded \;]
    =\E[\; \padExpr\;] \cdot \E\sbr{\IndPadded}
    = \E[\; \padExpr \;] \cdot p_j.
\end{align}
The last two equalities hold, resp., by independence of $\IndPadded$ and by Invariant 2.

First consider the case $i<j$. Then
\begin{align*}
\E\sbr{ \Lambda_{i,j}\cdot \IndPadded}
    & \geq \padding p_j
    &\EqComment{by \eqref{eq:lm:main-BIC:proof-padded} and the padded-BIC property of $\pi_j$}, \\
\E\sbr{ \Lambda_{i,j}\cdot \IndExplore}
    &\geq -\E\sbr{\IndExplore}\geq  -\padding\, p_j
    &\EqComment{a worst-case bound},\\
\E\sbr{\Lambda_{i,j}\cdot \IndExploit}
    &\geq 0
    &\EqComment{exploitation is always BIC}.
\end{align*}
Summing it up gives \refeq{eq:BICcondition}. For $i>j$, we have a similar argument, albeit for different reasons:
\begin{align*}
\E\sbr{ \Lambda_{i,j}\cdot \IndPadded}
    &\geq 0
    &\EqComment{by \eqref{eq:lm:main-BIC:proof-padded} and the BIC property of $\pi_j$}, \\
\E\sbr{ \Lambda_{i,j}\cdot \IndExplore}
    & = \E\sbr{(\mu_j- \mu_i)\cdot \IndExplore } \geq 0
    &\EqComment{by independence of $\IndExplore$},
\end{align*}
and
    $\E\sbr{\Lambda_{i,j}\cdot \IndExploit} \geq 0$
as before. We note that
    $\E\sbr{\Lambda_{i,j}\cdot \IndExploit}\geq 0$
would hold even if the exploitation phase would simply choose arm $1$; this is by independence of $\IndExploit$.

It remains to show the BIC property when the algorithm recommends some arm $\ell\neq j$, \ie that
    $\E\sbr{ \Lambda_{i,\ell}} \geq 0$
for all arms $i\neq \ell$. We derive
    $\E\sbr{ \Lambda_{i,\ell} \cdot \IndPadded} \geq 0$
and
    $\E\sbr{ \Lambda_{i,\ell} \cdot \IndExploit} \geq 0$
same way as before, and
    $\E\sbr{ \Lambda_{i,\ell} \cdot \IndExplore} = 0$
because the exploration phase always chooses arm $j$.
\end{proof}

Now, we show that a suitable policy $\pi_j$ exists, in terms of the following parameters:
\begin{align}
\padG &= \min_{j,\,q\in \Delta_{j-1}}\E\sbr{ \rbr{ \mu_j-\mu_q}_+ },
\quad \text{where }
    \mu_q := \textstyle  \sum_{i \in [K]} q_i\, \mu_i. \nonumber \\
\padN &=  \padC\, \padG^{-2}\, \log \rbr{\padG^{-1}},
    \label{eq:params-padded}
\end{align}
for large enough absolute constant $\padC$. In words, $\padG$ is the smallest ``expected advantage" of any arm $j$ over a convex combination of arms $i<j$. We merely guarantee existence of a policy via the minimax theorem, rather than specify how to compute it.

\begin{lemma}\label{lm:policy-existence}
For each arm $j\geq 2$ there exists a policy $\pi_j$ which is $(j,\padding,\padN)$-suitable,
for
 $\padding = \nicefrac{\padG}{10}$.
\end{lemma}
\begin{proof}[Proof Sketch]
Fix arm $j\geq 2$. Maximizing padding $\padding$ in \refeq{eq:padded-BIC-defn} is naturally expressed as a zero-sum game between a \emph{planner} who chooses a $(j,N)$-informed  policy $\pi_j$ and wishes to maximize the right-hand side of \eqref{eq:padded-BIC-defn}, and an \emph{agent} who chooses arm $i$; we call it the \emph{$(j,N)$-recommendation game}.%
\footnote{It is a finite game because there are only finitely many deterministic $(j,N)$-informed policies.}
 A policy $\pi_j$ is $(j,\padding)$-padded if and only if it is guaranteed payoff at least $\padding$ in this game. So, any maximin policy in this game is $(j,\padding)$-padded with $\padding = V_{j,N}$, where $V_{j,N}$ is the game value.

We connect the game value with $\padG$ via the minimax theorem, proving that
    $V_{j,N}\geq \padG/10$.
Indeed, using the minimax theorem and the linearity of expectation, we can write the game value as
\begin{align}\label{eq:pf:lm:policy-existence-1}
 V_{j,N}
    = \min_{q\in \Delta_{j-1}} \quad
    \max_{\text{$(j,N)$-informed policies $\pi_j$}}\quad
    \E\sbr{ (\mu_j - \mu_q)\cdot \ind{\pi_j=j}}.
\end{align}
If the policy $\pi_j$ knew the mean rewards $(\mu_1 \LDOTS \mu_j)$ exactly, it could recommend arm $j$ if and only if $\mu_j-\mu_q>0$. This would guarantee the minimax value of at least
    $G_j := \min_{q\in \Delta_{j-1}} \E\sbr{ \rbr{ \mu_j-\mu_q}_+ } $.
We show that for a large enough $N$ we can guarantee the minimax value of at least $\Omega(G_j)$. Specifically, for $N = \padN$ and any given distribution $q\in\Delta_{j-1}$ there exists a $(j,N)$-informed policy $\pi = \pi_{j,N}^{q}$ such that
\begin{align}\label{eq:pf:lm:policy-existence-2}
 \E\sbr{ (\mu_j - \mu_q)\cdot \ind{\pi=j}}
        \geq \tfrac{1}{10}\cdot \E\sbr{ \rbr{ \mu_j-\mu_q}_+ }.
\end{align}
This policy is very simple: given the data (the first $N$ samples of each arm $i\leq j$), recommend arm $j$ if and only if its empirical reward on this data is larger than the expected reward of distribution $q$ on the same data. For every given realization of the mean rewards $(\mu_1 \LDOTS \mu_j)$, we use Chernoff Bounds to compare
    $(\mu_j - \mu_q)\cdot \ind{\pi=j}$ and $(\mu_j - \mu_q)_+$,
which implies \eqref{eq:pf:lm:policy-existence-2}.


From here on, let us focus on the $(j,\padN)$-recommendation game. We have proved that there exists a $(j,\padN)$-informed policy $\pi_j$ that is $(j,\padding)$-padded, for $\padding = \padG/10$, even if we do not specify how to compute such a policy. However, we are not done yet, as we also need this policy to be BIC. First, we need the BIC property when arm $j$ is recommended. (As spelled out in \refeq{eq:BICcondition}; we already have this property for $i<j$, by the padded-BIC property, but we also need it for $i>j$.)

Any finite two-player zero-sum game has a minimax-optimal strategy that is \emph{non-weakly-dominated}: not weakly dominated by any other mixed strategy for the max player. So, let us take such a policy $\pi^*$. One can prove that any such policy is BIC when recommending arm $j$. The proof of this claim focuses on the probability of recommending arm $j$ as a function of the posterior mean rewards
    $\widetilde{\mu}_i := \E\sbr{ \mu_i \mid \mathcal \mG_{N,j}}$.%
    \footnote{Recall that $\mG_{N,j}$ denotes the $\sigma$-algebra generated by the first $N$ samples of each arm $i\leq j$.}
We show that the conditional probability
        $\Pr\sbr{ \pi_j = j \mid \widetilde{\mu}_1 \LDOTS \widetilde{\mu}_j }$
is non-decreasing in $\widetilde{\mu}_j$ and non-increasing in $\widetilde{\mu}_i$ for each arm $i<j$. This, in turn, allows us to invoke the FKG inequality and derive the claim.

Finally, we extend $\pi^*$ to a BIC policy. When $\pi^*$ does \emph{not} recommend arm $j$, we choose an arm $i$ for exploitation, \ie to maximize
     $\E\sbr{ \mu_i \mid \mG_{N,j}}$.
The resulting policy is BIC, hence $(j,\padding,\padN)$-suitable.
\end{proof}

\xhdr{Bootstrapping, revisited.}
The bootstrapping proceeds as follows, as per the pseudocode. We start with an exploitation phase with depth $N_0$. We choose $N_0$ so as to guarantee that we explore arm $j$ in the $\ZEROS_{j,N_0}$ event, \ie if all previous arms return $0$ rewards in the first $N_0$ samples. In the second phase, we explore arm $j$ with a small probability $p_j$; otherwise we do something else to guarantee incentives. Specifically, under event $\ZEROS_{j,N_0}$ we invoke a version of the padded phase; else we just exploit. This phase is BIC because the small probability of invoking the padded phase compensates for the exploration.

Given that the padded phase is now invoked under event $\ZEROS_{j,N_0}$, we need the padded-BIC property to hold conditional on this event. We transform policy $\pi_j$ into another policy $\pi = \FakePi$ which replaces the BIC and padded-BIC properties with conditional ones:
\begin{align}
\E\sbr{ (\mu_j - \mu_i)\cdot \ind{\pi=j} \mid \ZEROS_{j,N_0} }&\geq \padding
\quad\text{for each arm $i<j$},
\label{eq:padded-BIC-defn-zeroes}\\
\E\sbr{ (\mu_{\ell} - \mu_i)\cdot \ind{\pi=\ell} \mid \ZEROS_{j,N_0} }&\geq 0
\quad\text{for each arm $i\neq \ell$}.
\label{eq:ordinary-BIC-defn-zeroes}
\end{align}

This transformation is generic, in the sense that it works for any policy $\pi_i$ and any parameters.

\begin{restatable}{lemma}{transform}\label{lem:transform}
Fix padding $\padding>0$ and any parameters $N_0,N$. Given a $(j,\padding,N)$-suitable policy $\pi_j$, there exists a BIC policy $\pi = \FakePi$ which is $(j,N)$-informed and satisfies \eqref{eq:padded-BIC-defn-zeroes} and \eqref{eq:ordinary-BIC-defn-zeroes}.
\end{restatable}

\begin{proof}[Proof Sketch]
Intuitively, conditioning on arms $i<j$ being worse-than-usual should only help in finding a $(j,\lambda)$-padded BIC policy. To explicitly reduce to $\pi_j$, we consider the ``true data": the first $N$ samples of each arm $i<j$ which satisfy $\ZEROS_{j,N_0}$. We construct a ``fake" data set ($N$ samples for each arm $i<j$), which relates to the ``true" data in a certain way.
Specifically, the ``fake" posterior mean reward $\hat\mu_i$ for each arm $i<j$ is coupled with the ``true" posterior mean reward, denoted $\widetilde\mu_i$, so that $\hat\mu_i\geq \widetilde\mu_i$ almost surely and the distribution of $\hat\mu_i$ is the unconditional distribution of $\widetilde\mu_i$ (\ie the distribution of $\widetilde\mu_i$ without conditioning on the event $\ZEROS_{j,N_0}$). We define $\FakePi$ by applying $\pi_j$ on the fake data.

Since $\pi_j$ is $(j,\padding,N)$-suitable, we obtain the BIC (resp., padded-BIC) property against the fake posterior means $\hat\mu_i$, $i<j$. Since $\hat\mu_i\geq \widetilde\mu_i$, the corresponding properties hold against the values $\widetilde\mu_i$ as well.
\end{proof}

Now we prove that bootstrapping is BIC given a suitable policy $\pi_j$. The proof follows the same strategy as that of Lemma~\ref{lm:main-BIC}, but everything is conditioned on $\ZEROS_{j,N_0}$.

\begin{restatable}{lemma}{mainBIC}\label{lm:main-BIC0}
Consider Algorithm~\ref{alg:pyramidtight} with arbitrary parameters $N_0\leq N$ and $\padding>0$. Fix some arm $j\geq 2$. The bootstrapping phases are BIC as long as policy $\pi_j$ is $(j,\padding,N)$-suitable and parameter $N_0$ satisfies
\begin{align}\label{eq:cl:alg-obs:boot}
 \E\sbr{ \mu_j-\mu_i \mid \ZEROS_{j,N_0} } \geq 0
    \quad \text{for all arms $i<j$}.
\end{align}
\end{restatable}

\OMIT{ 
\begin{proof}
The proof follows the same strategy as that of Lemma~\ref{lm:main-BIC}. Line~\ref{line:initialexploit} is BIC since it is just exploitation. For the next phase, recommending any arm $i\neq j$ is BIC since it can only happen via exploitation. Hence we fix an arm $i\neq j$ and show that arm $j$ is BIC against arm $i$ in this phase. We define
$\IndPadded+\IndExploit+\IndExplore=1$ and $\Lambda_{i,j}$ as before. Note that we have probability $\Pr[\ZEROS_{j,N_0}]\cdot (1-p_j)$ to reach line~\ref{line:padded1}. Applying the guarantee of Lemma~\ref{lem:transform} conditional on this event, we obtain:

\begin{align*}\E[\Lambda_{i,j}\cdot \IndPadded]& \geq \Pr[\ZEROS_{j,N_0}]\cdot \lambda (1-p_j)\\
 &= p_j.\end{align*}

Moreover using a worst-case bound on exploration, and that exploitation is always BIC, we have:

\[\E[\Lambda_{i,j}\cdot \IndExplore]\geq -\E[\IndExplore]=-p_j\quad\text{ and }\quad  \E[\Lambda_{i,j}\cdot \IndExploit]\geq 0.\]

Adding, we obtain Equation~\ref{eq:BICcondition} for $i<j$. For $i>j$, similar reasoning again leads to the same three inequalities with $p_j$ replaced by $0$, i.e. $\E[\Lambda_{i,j}\cdot 1_E]\geq 0$ for $1_E\in \{\IndPadded,\IndExploit,\IndExplore\}.$ Altogether we have shown that playing arm $j$ is BIC against any arm $i$, establishing the BIC property for the bootstrapping phases.
\end{proof}
} 

\noindent\textbf{Putting the pieces together.}
Let us formulate an end-to-end guarantee for the algorithm in terms of the appropriate parameters. Recall that we use $\padG,\padN$ from \eqref{eq:params-padded} to handle the ``padded phase".
Given \eqref{eq:cl:alg-obs:boot}, we define two more parameters to handle the ``bootstrapping phase":
\begin{align}\label{eq:params-boot}
\bootN &= \min \cbr{ N_0\in\N:\;
    \E\sbr{ \mu_j-\mu_i \mid \ZEROS_{j,N_0}} > 0
    \; \text{for all arms $i<j$}}. \\
\bootP &= \min_{\text{arms $j$}}
    \Pr\sbr{ \ZEROS_{j,\bootN} }=\Pr\sbr{ \ZEROS_{K,\bootN} }.
\end{align}
In words, $\bootN$ is the smallest $N_0$ such that each arm $j$ is the best arm conditional on seeing $N_0$ initial samples from all arms $i<j$ that are all zeroes. Setting $N_0 = \bootN$ as we do in the theorem, $\bootP$ is the smallest phase-exploration probability after the bootstrapping phase.

\begin{restatable}{theorem}{goodalgo}
\label{thm:goodalgo}
Suppose algorithm \MainALG is run with parameters
    $N\geq \max(\bootN,\padN)$,
    $N_0 = \bootN$
and
    $\padding = \padG/10$.
Suppose each policy $\pi_i$ is $(j,\padding,N)$-suitable (such policies exist by Lemma~\ref{lm:policy-existence}). Then the algorithm is BIC and collects at least $N$ samples of each arm almost surely in time
\[ \RndsUB(N) = O\rbr{
    KN\;\padG^{-1}\; \log\rbr{\padG^{-1}\;\bootP^{-1}} }.\]
\end{restatable}

\begin{remark}\label{rem:init-strict-BIC}
\revedit{If all prior mean rewards $\E[\mu_i]$, $i\in [K]$ are pairwise distinct, then the algorithm in Theorem~\ref{thm:goodalgo} is strictly BIC: \refeq{eq:BIC-defn} is satisfied with a strict inequality, as per the same analysis.}
\end{remark}

Theorem~\ref{thm:goodalgo} provides an explicit formula for the time horizon of \MainALG. This formula is an \emph{upper bound} on the the sample complexity in question. It is polynomially optimal and allows for concrete corollaries, as we discuss in Section~\ref{sec:sample}. Setting $N\geq N_{\TS}$, we collect enough data to bootstrap Thompson Sampling, as per Section~\ref{sec:TS}.

How to compute a suitable policy $\pi_j$ for Lemma~\ref{lm:policy-existence}? For $K=2$ arms, one can take a very simple policy $\pi_2$: given the first $N$ samples of both arms, recommend an arm with a larger empirical reward. (This idea can be extended to an arbitrary $K$, but the padding $\padding$ degrades exponentially in $K$; we omit the easy details.) We provide a computationally efficient implementation for Beta priors in Section~\ref{sec:ext}. However, we do not provide a computational implementation for the general case.





\newcommand{\frsigma}{\nicefrac{1}{\sigma}}

\section{Sample Complexity of Incentivized Exploration}
\label{sec:sample}

We characterize the sample complexity of incentivized exploration: the minimal number of rounds need to collect $N$ samples of each arm. More precisely, we are interested in the smallest time horizon $T$ such that some BIC bandit algorithm collects $N$ samples of each arm almost surely; denote it $\RndsOPT(N)$. Note that $\RndsOPT(N)$ is determined by the joint prior, and does not depend on the mean rewards or observed rewards.

We are particularly interested in $\RndsOPT(1)$, the sample complexity of collecting at least one sample of each arm, and $\RndsOPT(N_\TS)$, the sample complexity of bootstrapping Thompson Sampling. The upper bound
    $\RndsUB(N_\TS)\geq \RndsOPT(N_\TS)$
comes from Section~\ref{sec:sampling}.%
\footnote{Since Theorem~\ref{thm:goodalgo} only provides $\RndsUB(N)$ for
    $N\geq \tilde{N} := \max(\bootN,\padN)$,
we define
    $\RndsUB(N) = \RndsUB(\tilde{N})$
for $N<\tilde{N}$.}
In this section, we derive a lower bound on $\RndsOPT(1)$ and mine the upper/lower bounds for a number of corollaries. In particular, we prove that upper and lower bounds are ``polynomially matching",
investigate how the sample complexity scales with $K$ and smallest variance
    $\sigma^2 = \min_{i\in[K]} \Var(\mP_i)$,
and work out canonical special cases. We focus on, and largely resolve, the distinction between polynomial and exponential dependence on $K$ and $\sigma^{-1}$.

\subsection{Lower bound}
\label{sec:sample-LB}

We provide a lower bound $\RndsLB\leq \RndsOPT(1)$, and prove that it matches $\RndsUB(N_\TS)$, in a specific sense that we explain below. The lower bound is driven by the following parameter:
\begin{align}
\label{eq:MainLB-defn}
\MainLB
    := \max_{j\in [K],\;q\in \Delta_{K}}
        \frac{\E\sbr{ (\mu_j-\mu_q)_-}}{\E\sbr{(\mu_j-\mu_q)_+}}
    =\max_{j\in [K],\;q\in \Delta_{K}}
        \frac{\mu_q^0-\mu_j^0}{\E\sbr{ (\mu_j-\mu_q)_+}}-1.
\end{align}

\begin{theorem}\label{thm:LB}
$\RndsOPT(1) \geq \RndsLB := \max(K,\,\bootN,\, \MainLB)$.
\end{theorem}

\begin{proof}
The lower bounds of $K,\bootN$ are clear so we focus on the last one. Suppose it is possible to sample all arms within $T$ rounds with some BIC algorithm. Fix an arbitrary arm $j$ and distribution $q\in \Delta_{K}$. In each rounds $t$ when arm $j$ is played, it must appear better than $q$, \ie
    $\E\sbr{ \ind{A_t=j}\cdot (\mu_j-\mu_q)}\geq 0$.
Summing this up over all rounds $t$, and letting $n_j$ be the number of times arm $j$ is played (which is a random variable) we obtain
    $\E\sbr{ n_j\cdot (\mu_j-\mu_q)}\geq 0$.
Now, the expression $n_j\cdot (\mu_j-\mu_q)$ is minimized by taking $n_j=1$ when $\mu_j<\mu_q$, and $n_j=T$ when $\mu_j\geq\mu_q$. Consequently we have
\[
    0
    \leq
    \E\sbr{ n_j\cdot (\mu_j-\mu_q)}
    \leq
    T\, \E\sbr{(\mu_j-\mu_q)_+} - \E\sbr{ (\mu_j-\mu_q)_-}.
\]
Rearranging shows that $T\geq \MainLB$ as claimed.
\end{proof}

\begin{remark}
\revedit{While Theorem~\ref{thm:LB} provides a rather tight estimate for $1$-sample complexity, we believe it does not imply any non-trivial lower bound on Bayesian regret of BIC $1$-sampling. To illustrate this point, we describe a stylized attempt to prove such lower bound via $\MainLB$, and explain why it fails.

Recall that BIC $1$-sampling must proceed for at least $\MainLB$ rounds by Theorem~\ref{thm:LB}. Fix arm $j$ that maximizes \refeq{eq:MainLB-defn}. Let us attempt to lower-bound Bayesian regret due to \emph{not} playing this arm, under an (overly pessimistic and unjustified) assumption that it is never played in the first $\MainLB$ rounds. Even then, the lower bound we obtain is at most $1$. 

Fix $q\in \Delta_K$ that maximizes \eqref{eq:MainLB-defn} for arm $j$, so that
\begin{align}\label{eq:attempted-LB-1}
     \MainLB
    =
    \frac{\E\sbr{ (\mu_j-\mu_q)_-}}{\E\sbr{(\mu_j-\mu_q)_+}}.
\end{align}
Assume $q_j=0$ w.l.o.g. (the quantities in \eqref{eq:MainLB-defn} are unchanged if we linearly rescale the remaining entries by $q_i\leftarrow \frac{q_i}{1-q_j}$ and set $q_j=0$). Each of the $\MainLB$ rounds when arm $j$ is not played contributes Bayesian regret 
\begin{align}\label{eq:attempted-LB-2}
    \E\sbr{\max_{i\in [K]}\mu_i-\max_{i\in [K]\setminus \{j\}} \mu_i}
    = \E\sbr{(\mu_j-\max_{i\neq j} \mu_i)_+} 
    \leq \E\sbr{ (\mu_j-\mu_q)_+}
    \leq 1/\MainLB.
\end{align}
(The first inequality in \refeq{eq:attempted-LB-2} holds because $\mu_q\leq \max_{i\neq j} \mu_i$ almost surely, which in turn holds because $q_j=0$. The second inequality in \refeq{eq:attempted-LB-2} holds by \eqref{eq:attempted-LB-1}.) Thus, this stylized argument lower-bounds Bayesian regret by $\MainLB$ times the left-hand side of \eqref{eq:attempted-LB-2}, which unfortunately is at most $1$.}

\end{remark}

\subsection{Polynomially matching upper/lower bounds}
\label{sec:sample-poly}

We express the upper bound $\RndsUB(N_{\TS})$ in terms of the parameters in the lower bound and the smallest variance $\sigma^2$. We conclude that the upper bound is polynomially optimal up to $\sigma^{-O(1)}$ factor.

\begin{corollary}\label{thm:matching}
Suppose the prior $\mP_i$ for each arm $i\in [K]$ has variance at least $\sigma^2$
. Then
\begin{align}\label{eq:thm:matching}
\RndsUB(N_{\TS})
    \leq \tildeO\rbr{ \sigma^{-4}\, K^{4.5}\, \MainLB^3\,\bootN+
        \sigma^{-2}\, K^{2.5}\, \MainLB\, \bootN^2}
    = \tildeO\rbr{ \sigma^{-4}\cdot \RndsLB^{O(1)}}.
\end{align}
\end{corollary}

This result suffices to resolve polynomial vs. exponential dependence on $\sigma^{-1}$ or $K$ or any other parameter. To make this statement explicit, suppose we have a family of problem instances indexed by a single parameter
    $\sigma^2 = \min_{i\in [K]} \Var(\mP_i)$,
and for this family the lower bound scales as $\RndsLB = f(\frsigma)$. Then
\begin{align}\label{eq:matching-resolve}
f(\frsigma) = \RndsLB \leq \RndsOPT(N) \leq \RndsUB(N_\TS) \leq \rbr{f(\frsigma)}^{O(1)}
\quad\forall N\in[N_\TS].
\end{align}
So, both upper and lower bounds are polynomial (resp., exponential) in $\frsigma$ when so is $f(\frsigma)$. A similar statement holds for any family of problem instances indexed by $K$ (or any other parameter), where $\sigma$ is lower-bounded by an absolute positive constant. In particular, assuming the lower bound scales as $\RndsLB = g(K)$, we have
\eqref{eq:matching-resolve} with $f(\frsigma)$ replaced by $g(K)$.

\subsection{Dependence on the number of arms}
\label{sec:sample-K}

To investigate how $\RndsOPT(\cdot)$ scales with $K$, we need to separate the dependence on $K$ from the dependence on the priors $\mP_i$, $i\in[K]$. Therefore,
we posit that all priors come from some (possibly infinite) collection $\mC$;
such problem instances are called $\mC$-consistent. Keeping $\mC$ fixed, we study the dependence on $K$ in the worst case over all $\mC$-consistent instances.

We find a curious dichotomy: under a mild non-degeneracy assumption, either
    $\RndsUB(N_\TS)= O_{\mC}(K^3)$
for all $\mC$-consistent instances (under a fairly reasonable assumption on $\mC$), and otherwise
    $\RndsLB(1)> \exp\rbr{\Omega_{\mC}(K)}$
for some $\mC$-consistent instance. In its simplest form, this dichotomy can be stated for a finite collection $\mC$. For each prior $\mP\in \mC$, we consider the supremum of its support,
    $\sup(\mP) := \sup(\text{support}(\mP))$.
We relate these quantities to the largest prior mean reward over $\mC$, denoted
    $\Phi_\mC = \sup_{\mP\in \mC}\; \E[\mP]$.

\begin{theorem}\label{thm:K-simple}
Let $\mC$ be a finite collection of distributions over $[0,1]$. Then
\begin{OneLiners}
\item[(a)] If $\sup(\mP)> \Phi_\mC$ for all priors $\mP\in\mC$,
then $\RndsUB(N_\TS)= O_{\mC}(K^3)$
for all $\mC$-consistent instances.

\item[(b)] If $\sup(\mP)< \Phi_\mC$ for some $\mP\in\mC$,
then $\RndsLB(1)> \exp\rbr{\Omega_{\mC}(K)}$
for some $\mC$-consistent instance.
\end{OneLiners}
Either (a) or (b) holds assuming that
    $\min_{\mP\in\mC}\,\sup(\mP)\neq \Phi_\mC$.
\end{theorem}

The assumption in part (a) is very reasonable. For instance, it holds whenever each prior has a positive density everywhere on the $[0,1]$ interval.

A quantitative version of Theorem~\ref{thm:K-simple}, which is the version we actually prove, is more difficult to state. For a given parameter $\delta>0$, the easy vs. hard distinction is as follows:
\begin{align}
\text{$\mC$ is called \emph{$\delta$-easy} if}\quad
    &\inf_{\mP\in \mC}\;
        \E_{\mu\sim \mP}\sbr{ (\mu - \Phi_\mC)_+}>\delta
        \label{eq:K-pwise-inv} \\
\text{$\mC$ is called \emph{$\delta$-hard} if}\quad
    & \inf_{\mP\in \mC}\;
\Pr_{\mu\sim \mP}\sbr{ \mu \geq  \Phi_\mC-\delta}=0.
    \label{eq:K-pwise-inv2}
\end{align}
Moreover, we need a quantitative version of pairwise non-dominance \eqref{eq:pwise-cond}: $\mC$ is called \emph{$\delta$-non-dominant} if
\begin{align}\label{eq:K-pwise-quant}
\E\sbr{ (\mu_j^0-\mu_i)_+}\geq \delta
    \quad\text{for every $\mC$-consistent instance}.
\end{align}
In terms of these properties, the characterization extends to infinite collections $\mC$. Note that if $\mC$ is $\delta$-easy (resp., $\delta$-non-dominant) then so is any subset of $\mC$. Likewise, if $\mC$ is $\delta$-hard, then so is any superset of $\mC$.

\begin{theorem}\label{thm:K-general}
Let $\mC$ be a (finite or infinite) collection of distributions over $[0,1]$.
 Then

\begin{OneLiners}
\item[(a)] If $\mC$ is $\delta$-easy and $\delta$-non-dominant,
then $\RndsUB(N_\TS)= \tildeO\rbr{K^3/\delta^4}$
for all $\mC$-consistent instances.

\item[(b)] If $\mC$ is $\delta$-hard,
then $\RndsLB(1)> \exp\rbr{\Omega_{\delta}(K)}$
for some $\mC$-consistent instance.
\end{OneLiners}
If $\mC$ is finite, then it is either $\delta$-easy or $\delta$-hard for some $\delta>0$, provided that
    $\min_{\mP\in\mC}\,\sup(\mP)\neq \Phi_\mC$.
\end{theorem}

\begin{proof}
For part (a) we upper-bound in Lemma~\ref{lem:easyexploreparameters} the prior-dependent parameters as follows:
    $N_{\TS}=\tildeO(K\delta^{-2})$
for Thompson Sampling,
    $\padG\geq \delta$ and $\padN=\tildeO(\delta^{-2})$
for the padded phase, and
    $\bootN=\tildeO(\delta^{-1})$
and
    $\log(\bootP^{-1})=\tildeO(K\delta^{-1})$
for the bootstrapping phase. Then part (a) follows from Theorem~\ref{thm:goodalgo}.

For part (b), assume $\mC$ is $\delta$-hard. Then for any $\eta>0$ there exist priors $\mP,\mP'\in\mC$ with
\[\Pr_{\mu\sim \mP}\sbr{ \mu\leq \E[\mP']-\nicefrac{\delta}{2}}\geq 1-\eta.\]
Consider the problem instance in which
    $\mu_1,\dots,\mu_{K-1}\sim \mP'$ and $\mu_K\sim \mP$.
Let $q$ be the uniform distribution on $[K-1]$. 
Using Chernoff Bounds, we have
\[\Pr\sbr{ \mu_q\leq \E[\mP']-\nicefrac{\delta}{2}} \leq e^{-\Omega(K\delta^{2})}.\]
Consequently,
    $\E\sbr{ (\mu_K-\mu_q)_+ } \leq \Pr\sbr{ \mu_K\geq \mu_q} \leq 2\eta$.

Moreover, it holds that
    $\E[\mu_q-\mu_K]\geq\nicefrac{\delta}{2}-\eta$.
Taking
    $\eta\leq \min\rbr{ \nicefrac{\delta}{4}\,,\,e^{-\Omega(K\delta^{2})}}$,
we conclude that
\[\RndsLB\geq \MainLB
    \geq \frac{\E[\mu_q-\mu_K]}{\E[(\mu_K-\mu_q)_+]}-1
    \geq \nicefrac{\delta}{4}\cdot e^{-\Omega(K\delta^{2})}-1. \qedhere\]
\end{proof}


\OMIT{ 
    In all cases we assume that

    \[\mathbb E[(\mu_j^0-\mu_i)_+]\geq \delta\]

    for all $\mu_i,\mu_j\sim \mP_i,\mP_j\in \mC$. Up to the uniform value of $\delta$ this asserts that all arms can be explored. The positive condition, that $\mC$ is $\delta$-easy to explore, is as follows:

The opposite condition, that $\mC$ is $\delta$-hard to explore, states:

\begin{align}\label{eq:K-pwise-inv2}
\inf_{\mP\in \mC}\quad
\Pr_{\mu\sim \mP}\sbr{ \mu \geq  \Phi_\mC-\delta}=0.
\end{align}


When $|\mC|<\infty$, one of the above holds unless

\[\inf_{\mP} \sup(\text{support}(\mP))=\Phi_\mC.\]

For fully general $\mC$, one of the above conditions holds for some $\delta>0$ unless we have the equality:

\[\lim_{\eps\to 1}\inf_{\mP} F_{\mP}^{-1}(1-\eps)=\Phi_\mC.\]

Here $F_{\mP}^{-1}$ is the inverse cumulative density function for $\mP$. If a strict inequality holds then $\mC$ is either $\delta$-easy or $\delta$-hard for some constant $\delta>0$. As such we consider families $\mC$ which are neither easy nor hard to be edge cases.

In particular,
    $\Pr[\mu_i > \mu_j^0]>0$,
\ie any arm $i$ can be better than the prior mean of any other arm $j$. Recall that the pairwise non-dominance condition \eqref{eq:pwise-cond} is similar but ``inverted": any arm $i$ can be \emph{worse} than the prior mean of $j$. Let us spell out a version of the latter condition for $\mC$:
\begin{align}\label{eq:K-pwise}
\inf_{\text{priors }\mP\in \mC}\quad
\Pr_{\mu\sim \mP}\sbr{ \mu < \Phi^{\inf}_\mC}>0,
    \quad\text{where}\quad
\Phi^{\inf}_\mC = \inf_{\text{priors }\mP\in \mC}\; \E[\mP].
\end{align}
Both conditions can be ``easily" satisfied simultaneously: they merely assert that the prior mean rewards lie in some interval
    $\sbr{\Phi^{\inf}_\mC,\, \Phi_\mC}$
with non-empty ``tails"
    $[0,\Phi^{\inf}_\mC)$
and
    $(\Phi_\mC,\,1]$,
and each prior assigns probability at least $\eps>0$ to both ``tails". Thus, the formal statement is as follows.

\begin{theorem}\label{thm:dichotomy-K}
Let $\mC$ be an arbitrary (and possibly infinite) collection of distributions over $[0,1]$ which is $\delta$-easy to explore. If $\mC$ satisfies \eqref{eq:K-pwise-inv}, then
    $\RndsUB(N_\TS)= \tilde O(K^3\delta^{-4})$
for all $\mC$-consistent instances. On the other hand if $\mC$ is $\delta$-hard to explore, there exists a $\mC$-consistent instance such that
    $\RndsLB(1)> \exp\rbr{\Omega_{\delta}(K)}$.

\end{theorem}

\begin{proof}

The first case follows from Lemma~\ref{lem:easyexploreparameters} stated just below, which bounds all the fundamental problem parameters when $\mC$ is $\delta$-easy to explore. Applying the estimate of Theorem~\ref{thm:pyramidtightfast} completes the proof.

For the exponential lower bound, if $\mC$ is $\delta$-hard to explore, then for any $\eta>0$ there exist $\mP,\mP'\in\mC$ with \[\Pr_{\mu\sim \mP}[\mu<\E_{\mu'\sim\mP'}[\mu']-\frac{\delta}{2}]\geq 1-\eta.\] We take $\mu_1,\dots,\mu_{K-1}\sim \mu$ and $\mu_K\sim \mu'$. Letting \[q=\left(\frac{1}{K-1},\frac{1}{K-1},\dots,\frac{1}{K-1}\right)\in\Delta_{K-1}\] a Chernoff estimate shows that

\[\mathbb P[\mu_q\leq c_U-\frac{\delta}{2}]\leq e^{-\Omega_{\delta}(K)}+\eta.\]

Moreover we clearly have $\E[\mu_q-\mu']\geq \frac{\delta}{2}$. Since $\eta$ was arbitrary (and in particular may decay as $K\to\infty$) we conclude that the lower bound $L$ is exponential in $K$ as desired.
\end{proof}

\begin{restatable}{lemma}{easyexploreparameters}
\label{lem:easyexploreparameters}
If $\mC$ is $\delta$-easy to explore, then:
\begin{itemize}
    \item $N_{\TS}=\tilde O(K\delta^{-2})$.
    \item $G\geq \delta$.
    \item $\padN=\tilde O(\delta^{-2}).$
    \item $\bootN=\tilde O(\delta^{-1})$.
    \item $\log(\Pexplore^{-1})=\tilde O(K\delta^{-1})$.
\end{itemize}
\end{restatable}

} 






\subsection{Canonical priors}
\label{sec:sample-canon}

We consider two canonical examples of incentivized exploration: when the priors $\mP_i$ are truncated Gaussians and when they are Beta  distributions. We find that the optimal sample complexity $\RndsOPT(N_{\TS})$ scales polynomially on $K$ and exponentially in the ``strength of beliefs".

For truncated Gaussian priors, we focus on the case when all Gaussians have the same variance $\sigma^2$, and we find that the sample complexity is polynomial in $K$ and exponential in $\sigma^{-2}$. Thus, \emph{strong beliefs}, as expressed by small variance $\sigma^2$, is what makes incentivized exploration difficult.

\begin{restatable}{corollary}{gaussianbound}
\label{thm:gaussianbound}
Let $\widetilde N(\nu,\sigma^2)$ be a Gaussian with mean $\nu$ and variance $\sigma^2\leq 1$, conditioned to lie in $[0,1]$. Suppose $\mP_i\sim \widetilde N(\nu_i,\sigma^2)$ for each arm $i$, where
    $\nu_1 \LDOTS \nu_K\in [0,1]$.
Then
\begin{OneLiners}
\item[(a)] $\Tbic{N_{\TS}} = K^3\cdot \poly\rbr{\sigma^{-1},\,e^{R^2}}$
where
    $R=\sigma^{-1}\;\max_{i,j} |\nu_i-\nu_j|$.
\item[(b)]
    $\RndsLB \geq e^{\Omega(1/\sigma^2)}$
when
    $\max_{i,j} |\nu_i-\nu_j|$
is a positive absolute constant.
\end{OneLiners}
\end{restatable}

\newcommand{\strength}{\term{strength}}


Strength of beliefs expressed by a particular truncated Gaussian can be usefully interpreted as the number of samples inherent therein. Indeed, any Gaussian distribution with variance $\sigma^2$ can be represented as a Bayesian update of a unit-variance Gaussian given $M = \Theta(\sigma^{-2})$ independent data points. Thus, the sample complexity in Corollary~\ref{thm:gaussianbound} is exponential in $M$.

We obtain a similar result for Beta priors. We define the \emph{strength} of distribution $\mP= \term{Beta}(a,b)$ as the sum $a+b$. When $a,b\in\N$, the number of data points needed to obtain $\mP$ as a posterior starting from a uniform prior $\term{Beta}(1,1)$ is $a+b-2$;
we interpret it as the strength of beliefs expressed by the prior. Note that $M=\strength(\mP)$ approximately captures variance: indeed,
        $1/\Var(\mP) \leq [M,\,M^2]$,
A problem instance is called \emph{$M$-strong} if $\term{strength}(\mP_i)\equiv M$.
We prove that $\RndsUB(N_{\TS})$ is polynomial in $K$ for any fixed $M$, and exponential in $M$ even for $K=2$; the latter dependence is inevitable.

\begin{restatable}{corollary}{betabound}
\label{thm:betabound}
Suppose all priors $\mP_1 \LDOTS \mP_K$ are Beta distributions.
\begin{OneLiners}
\item[(a)]
    $ \RndsUB(N_{\TS}) \leq K^3\cdot \rbr{\min(K,M)}^{O(M)}$
if $\strength(\mP_i)\leq M$ for all arms $i$.

\item[(b)]
    $\RndsLB\geq   \rbr{\min(K,M)}^{\Omega(M)}$
for \textbf{some} $M$-strong problem instance.

\item[(c)]
    $\RndsLB\geq 2^{\Omega(M)}$
for \textbf{any} $M$-strong problem instance such that
    $\mu_1^0-\mu_K^0 \geq\Omega(1)$.\\
More generally, this holds whenever arms $i\neq j$ have strength at least $M$ and
        $|\mu_i^0-\mu_j^0| \geq \Omega(1)$.
\end{OneLiners}
\end{restatable}

The lower bound in part (b) holds if all arms $i<K$ have the ``smallest" prior
    $\mP_i= \term{Beta}(M-1,1)$,
and arm $K$ has the ``largest" prior
$\mP_K\sim \term{Beta}(1,M-1)$.

In fact, (slightly weaker versions of) Corollaries~\ref{thm:gaussianbound}(a) and~ \ref{thm:betabound}(a) can be derived from Theorem~\ref{thm:K-general}. This is because the corresponding collections of priors are $\delta$-easy and $\delta$-non-dominant.

\begin{restatable}{lemma}{examplescollections}\label{lm:examples-collections}
The collection of all truncated Gaussians $\widetilde N(m,\sigma^2)$, $m\in[0,1]$ with a fixed variance $\sigma^2<1$ is $\delta$-easy and $\delta$-non-dominant with
$\delta = e^{-\Omega(1/\sigma^2)}$. Likewise, the collection of all Beta distributions of strength at most $M\geq 1$ is $\delta$-easy and $\delta$-non-dominant with
    $\delta = M^{-O(M)}$.
\end{restatable}

\subsection{One well-known arm}
\label{sec:sample-wellknown}

Let us consider an important special case when some arm $\ell$ represents a well-known, ``default" alternative, and the other arms are new to the agents. Put differently, agents have strong beliefs on one arm but not on all others. We would like to remove the dependence on the well-known arm as much as possible. We obtain two results of this flavor: a version of Corollary~\ref{thm:matching} on polynomially matching upper/lower bounds, and a version of Corollary~\ref{thm:betabound} on Beta priors.
We strengthen Corollary~\ref{thm:matching} under a mild non-degeneracy condition which ensures that the priors are not extremely concentrated near $1$.

\begin{corollary}\label{eq:thm:matching-improved}
\refeq{eq:thm:matching} holds if
    $\Var(\mP_i)\geq \sigma^2$
for all but one arm $i$, provided that
\begin{align}\label{eq:non-degeneracy}
\Pr\sbr{ \mu_i \leq 1-\sigma } \geq e^{-1/\sigma}
    \quad\text{for all arms $i$}.
\end{align}
\end{corollary}

We obtain a stronger result focusing on Beta priors. Then the exponential dependence on $\strength(\mP_i)$ holds only for arms $i\neq \ell$, whereas the dependence on $\strength(\mP_\ell)$ is only polynomial. There is also a ``penalty term" which scales exponentially in $\ell$; this is mild if $\ell$ is small.




\begin{restatable}{corollary}{realistic}
Suppose all priors $\mP_1 \LDOTS \mP_K$ are Beta distributions. Suppose $\strength(\mP_\ell)=M$ for some arm $\ell$, and $\strength(\mP_i)\leq m$ for all other arms $i$, where $M\geq m\geq 2$.
Then
\[
    \RndsUB(N_{\TS})
        \leq  M^2\cdot K^{O(1)}\cdot \max\rbr{ m,\;(1-\mu_\ell^0)^{-1},\;(\mu_\ell^0)^{-(\ell-1)}}^{O(m)}.
\]
In particular,
    $ \RndsUB(N_{\TS}) \leq M^2\cdot K^{O(1)}\cdot m^{O(\ell\, m)}$
if $\mu_\ell^0\in [\nicefrac{1}{4},\, \nicefrac{3}{4}]$.
\end{restatable}


While dependency on \emph{one} well-known arm can be mitigated, the lower bound in Lemma~\ref{lm:examples-collections}(c) rules out a similar improvement if $\strength(\mP_i)\geq M$ for two or more arms $i$.

We obtain a particularly clean characterization for $K=2$ arms, which is worth stating explicitly.

\begin{corollary}
Assume $K=2$ arms and Beta priors
    $\mP_i\sim \term{Beta}(a_i,b_i)$.
Let
    $M = \max_i(a_i+b_i)$ and $m = \min_i(a_i+b_i)\geq 2$.
Suppose
    $\mu_i\in [\eps,1-\eps]$
for some $\eps>0$ and both arms $i$. Then
\[ \RndsUB(N_{\TS})
        \leq  M^2\cdot \max\rbr{m,\nicefrac{1}{\eps}}^{O(m)}.
\]
Moreover,
    $\MainLB\geq 2^{\Omega(m)}$
provided that
 $\mu_1^0-\mu_2^0 \geq \Omega(1)$.
\end{corollary}

\section{Extensions via Improved Algorithms}
\label{sec:ext}

\vspace{-2mm}
\bigXhdr{Improved Algorithm for ``Easy" Problem Instances.}
We design a new algorithm for collecting $N$ samples of each arm, for any given $N$  (Algorithm~\ref{alg:pyramidsuperfast} in Appendix~\ref{sec:fastest-algo}). This algorithm is tailored to the ``easy" case of the polynomial vs. exponential dichotomy in Section~\ref{sec:sample-K}. Formally, a problem instance is called \emph{$\delta$-easy}, $\delta>0$ if the collection of  priors $(\mP_1 \LDOTS \mP_K)$ is $\delta$-easy and $\delta$-non-dominant. We obtain linear dependence on $K$, which is obviously optimal for a fixed $N$, along with computational efficiency.


\begin{restatable}{theorem}{linearexplore}
\label{thm:linearexplore}
Given a $\delta$-easy problem instance with $K$ arms, Algorithm~\ref{alg:pyramidsuperfast} is BIC and collects at least $N$ samples of each arm almost surely in
    $\tildeO\rbr{\frac{KN}{\delta}+\frac{K}{\delta^4}}$
rounds, for any desired $N\in\N$. The running time for each round is $O(1)$ plus one call to ``exploitation" given up to $\tilde O(\delta^{-2})$ samples per arm.
\end{restatable}

Thus, going back to the setup in Section~\ref{sec:sample-K}, the dependence on $K$ admits a crisp \emph{linear vs. exponential dichotomy} for an arbitrary collection $\mC$. Corollaries for Beta and truncated Gaussian priors, also with linear dependence on $K$, follow from Lemma~\ref{lm:examples-collections}. In particular for strength-$M$ Beta priors, using Corollary~\ref{cor:tslogk}  shows that we efficiently achieve $\RndsOPT(N_{\TS})\leq\tilde O(K\cdot M^{O(M)}).$

The main insight behind Algorithm~\ref{alg:pyramidsuperfast} is that for $\delta$-easy instances the ``suitable" policy $\pi_j$ does not need to depend on the samples from arms $i<j$. We choose an arm $j_0$ uniformly at random and explore it almost surely like in \MainALG, but using policy $\pi_j$ as above. Afterwards we randomize between a padded phase for arm $j_0$ or an exploration phase for a randomly chosen arm $i\neq j_0$. This algorithm is BIC despite only going through the exponential growth process for a single arm $j_0$. It is efficient because we are able to let $\pi_j$ be a form of exploitation.

\bigXhdr{Fine-tuning the main algorithm.}
Algorithm~\ref{alg:pyramidtight} is somewhat wasteful when many samples are desired, i.e. $N\gg \max(\padN,\bootN)$. Indeed, only $\bootN$ samples are needed for the bootstrapping phase to be BIC, and only $\padN$ samples are needed for each iteration of the \term{while} loop to collect $\padN$ samples of arm $j$ almost surely. After that, the remaining samples for arm $j$ can be collected more efficiently: indeed, one can directly randomize between an exploration phase and a padded phase. As we show in Appendix~\ref{sec:improved-algo}, these modifications reduce the number of rounds for collecting $N$ samples of each arm to
\begin{align}\label{eq:alg2-guarantee}
\RndsUBtwo(N) =
O\rbr{
    K\;\padG^{-1}\; \rbr{\padN\,\log(\padG^{-1}\,\bootP^{-1}) + \bootN + N }}.
\end{align}


\bigXhdr{Efficient Computation for Beta Priors.}
We also present a computationally efficient version of the main algorithm. We focus on the special case of Beta priors of strength at most $M$, and recover the statistical guarantee in Corollary~\ref{thm:betabound}(a).
The bottleneck is to compute a $(j,\padding,N)$-suitable policy $\pi_j$, as all other steps are computationally efficient.\footnote{The construction of $\FakePi$ in Lemma~\ref{lem:transform} is by computationally efficient reduction to $\pi_j$, so this causes no issues.}
Such policy can then be plugged into Algorithm~\ref{alg:pyramidtight}. The details can be found in Appendix~\ref{sec:computation-beta}.

The key idea is that if we stochastically increase the priors for
    $\mu_1 \LDOTS \mu_{j-1}$
and keep the padding $\padding$ fixed, this only makes our task more difficult. (Formalizing this involves a coupling argument between ``true" and ``artificial" data.)
Thus, we reduce the problem to one in which all priors are $\term{Beta}(1,M)$. Now that all arms $i\in [j-1]$ are i.i.d., a symmetry argument shows that the only convex combination $q$ of these arms that we need to consider in the recommendation game (see \refeq{eq:pf:lm:policy-existence-1} in the proof of Lemma~\ref{lm:policy-existence}) is the unweighted average. Consequently, the problem of finding a maximin policy for the $j$-recommendation game reduces to competing against this unweighted average, which can be done via a simple comparison.


\begin{restatable}{theorem}{betaefficient}
\label{theorem:betaefficient}
Fix arm $j$, the number of arms $K$, parameters $N,M\in\N$ and padding $\padding>0$. Suppose there exists a $(j,N)$-informed policy $\pi_j$ which is BIC and $(j,\padding)$-padded BIC for all problem instances with $K$ arms and Beta priors of strength at most $M$. Then there is policy $\pi_j^{\term{eff}}$ with these properties which can be computed efficiently, namely in time $\poly(K,M,N)$. In particular, one can take $\padding=(\min(K,M))^{O(M)}$. Plugging these $\pi_j^{\term{eff}}$ and $\padding$ into Algorithm~\ref{alg:pyramidtight} yields the guarantee in Corollary~\ref{thm:betabound}(a).
\end{restatable} 

\newpage

\addcontentsline{toc}{section}{References}
\bibliographystyle{ACM-Reference-Format}
\begin{small}
\bibliography{bib-abbrv,bib-AGT,bib-bandits,bib-ML,bib-random,bib-slivkins,bib-competition,bib-Sellke}
\end{small}

\appendix

\newpage

\addtocontents{toc}{\protect\setcounter{tocdepth}{2}}
\tableofcontents
\newpage

\section{Explorability Characterization}
\label{sec:explorability}

We prove that the pairwise non-dominance condition \eqref{eq:pwise-cond} is in fact necessary, and use this condition to characterize which arms can be explored. Our result is modulo a minor non-degeneracy assumption: an arm is called \emph{support-degenerate} if its true mean reward (according the the prior) is always in the set $\{x,1\}$ for some $x\in [0,1)$. The proof is relatively simple: all the ``heavy lifting" is done in the algorithmic results.

\begin{theorem}\label{thm:explorability}
Suppose all arms are not support-degenerate, in the sense defined above.
Let $S$ be the set of all arms $i$ which satisfy the pairwise non-dominance condition \eqref{eq:pwise-cond}. 
\begin{itemize}
\item[(a)] Only arms in $S$ can be explored. More formally: for each arm $i$, if there exists a BIC algorithm which explores this arm with a positive probability, then $i\in S$.

\item[(b)] All our algorithms can be restricted to $S$. More formally: if an algorithm is guaranteed to be BIC under \eqref{eq:pwise-cond}, then this algorithm remains BIC if it is restricted to the arms in $S$.

\item[(c)] All arms in $S$ can be explored. More formally: there is a BIC algorithm which explores all arms in $S$ with probability $1$ within some finite time $t^*)$ depending on the prior.

\end{itemize}
\end{theorem}

\begin{proof}
First we show that any arm $i\not\in S$ is not explorable. Let arm $j$ be a ``reason" why $i\not\in S$, \ie suppose $\mu_i^0\leq m_j$ for some arm $j\neq i$, where $m_j$ is the minimum value in the support of $\mu_j$. For sake of contradiction, let $t$ be the first time at which arm $i$ is explored with positive probability by a BIC algorithm. Then
$\E^t[\mu_i]=\mu_i^0 \leq m_j<\E^t[\mu_j]$, contradiction.

The last inequality is strict by non-degeneracy:  since arm $j$ is not support-degenerate, its posterior reward distribution cannot be a point mass at $m_j$. Indeed, the only values of $\mu_j$ that can be ruled out with probability $1$ in finite time are $0$ and $1$ (because of Bernoulli rewards). Thus, if the posterior on $\mu_j$ is a point mass on $m_j$, then $\mu_j$ must have support in $\{0,m_j,1\}$. If the support includes $0$, then $m_j=0$ by definition, and the support in $\{m_j,1\}$ is ruled out by non-degeneracy.
This proves part (a).

Parts (b) and (c) easily follow. For part (b), consider the algorithm restricted to the arms in $S$. If the BIC condition is ever violated, it can only be because some arm $j\not\in S$ would be preferred by an agent. Thus, we obtain a BIC algorithm which explores arm $j$ with positive probability (by recommending arm $j$ for the same history), which contradicts part (a). Part (c) follows by applying part (b) to Algorithm~\ref{alg:pyramidtight}.
\end{proof}

\section{Tools from Probability}
\label{app:tools}

\subsection{Fortuin-Kasteleyn-Ginibre (FKG) inequality for correlation}
\label{app:mon-ineq}


\begin{lemma}{\cite[Theorem 6.2.1]{ProbabilisticMethod}}[FKG Inequality]
\label{lem:FKG}
Consider measures $\nu_1 \LDOTS \nu_n$ on $\mathbb R$, and let $\nu = \prod_{i\in[n]} \nu_i$ be the product measure on $[0,1]^n$. Suppose $f,g:\mathbb R^n\to [0,1]$ are functions which are increasing in each coordinate. Then
    $\E[f\cdot g]\geq \E[f]\cdot \E[g]$,
where $\E$ denotes expectation relative to $\nu$. Likewise, if $f$ is coordinate-wise increasing and $g$ is coordinate-wise decreasing then
$\E[f\cdot g]\leq \E[f]\cdot \E[g]$.
\end{lemma}

\begin{remark}\label{rem:fkgmixed}

In fact if $f,g$ are both increasing or both decreasing in each coordinate, then the conclusion above still holds as we can simply negate some coordinates in the parametrization, i.e. view $f,g$ as increasing functions of $-x_i$ instead of decreasing functions of $x_i$. We will refer to this as the \emph{mixed-monotoncity} FKG inequality to highlight the slight subtlety in its application.

\end{remark}


\begin{corollary}\label{lem:1varFKG}
Suppose $f,g:\mathbb R\to [0,1]$ are increasing, and $X$ is a random variable. Then
    $\E[f(X)\,g(X)]\geq \E[f(X)]\cdot\E[g(x)].$
If $f$ is increasing and $g$ is decreasing then $\E[f(X)\,g(X)]\leq \E[f(X)]\cdot\E[g(X)].$
\end{corollary}

The above corollary is sometimes known as the Chebyshev inequality. Not to conflate it with the better known Chebyshev inequality from probability theory, we will also call it the FKG inequality in this work.

\begin{corollary}\label{cor:staticFKG}
Let $\mathcal G$ be a static $\sigma$-algebra and $F$ a non-negative $\mathcal G$-measurable function which is increasing in the posterior mean of each arm $a_s$ for $s\in S\subseteq [K]$. Then $F$ is positively correlated with any coordinate-wise increasing function of the true values $(\mu_s)_{s\in S}$.
\end{corollary}

\begin{proof}
If $\mu_s$ increases, this stochastically increases the (binomially distributed) empirical mean of arm $a_s$, hence increases the expectation of $F$. Hence $\E[F|(\mu_s)_{s\in S}]$ is a coordinate-wise increasing function, and therefore by FKG is positively correlated with any other coordinate-wise increasing function of $(\mu_s)_{s\in S}$.
\end{proof}

\subsection{Bayesian concentration: proof of Lemma~\ref{lem:TSchernoff}}
\label{app:Bayesian-concentration}

\begin{lemma}
\label{lem:freqguarantee}
Suppose there exists an estimator $\theta=\theta(\gamma)$ for a parameter $\mu$ given an observation signal $\gamma$ which satisfies a concentration inequality
    $  \Pr\sbr{ |\theta-\mu|\geq \eps }\leq\delta$.
Assume we start with a prior $P$ over $\mu$ before observing $\gamma$. Then if $\mu$ is in fact chosen according to $P$ and $\hat\mu$ is a sample from the posterior distribution for $\mu$ conditional on $\gamma$, we have
    $\Pr\sbr{ |\hat\mu-\mu|\geq 2\eps}\leq 2\delta$.
\end{lemma}

\begin{proof}
By assumption,
    $\Pr\sbr{ |\theta-\mu|>\eps}\leq\delta$.
Also, $(\mu,\theta)$ and $(\hat\mu,\theta)$ are identically distributed; choosing $\hat\mu$ amounts to resampling $\mu$ from the joint law of $(\mu,\theta)$. Therefore,
    $\Pr\sbr{ |\theta-\hat \mu|>\eps}\leq\delta$.
Combining these two inequalities gives the Lemma by triangle inequality.
\end{proof}


\begin{proof}[Proof of Lemma~\ref{lem:TSchernoff}]
For simplicity show the claimed inequalities for $\mu_i$, explaining the generalization to weighted averages $\mu_q$ at the end. First we apply Lemma~\ref{lem:freqguarantee} with $\theta_i$ being the empirical mean from the first $\varepsilon^{-2}$ samples of arm $i$ - this is $\mathcal F_t$ measurable by assumption. With this choice of estimator $\theta_i$, for any true mean $\mu_i$ the ordinary Chernoff bound implies $\Pr[|\theta_i-\mu_i|\geq r\varepsilon]\leq C_0e^{-r^2/C_0}$ for some absolute constant $C_0$. Then \eqref{eq:BChernoff} follows from Lemma~\ref{lem:freqguarantee}.

We now deduce the second claim \eqref{eq:BChernoff-2} from the first. Indeed \eqref{eq:BChernoff} combined with the tail-decay definition of sub-Gaussianity implies that $\frac{\hat\mu_i-\mu_i}{\varepsilon}$ is $O(1)$ sub-Gaussian (averaged over all randomness in the problem). Moreover it manifestly has mean $0$. 
Also,
\[
    \E[\mu_i|\mF_t]-\mu_i=\E\sbr{\hat\mu_i-\mu_i\mid\sigma(\mF_t,\mu_i)}.
\]
Using the exponential moment definition of sub-Gaussianity, the fact that $\frac{\E[\mu_i|\mF_t]-\mu_i}{\varepsilon}$ is the conditional expectation of an $O(1)$ sub-Gaussian variable implies it is also $O(1)$ sub-Gaussian. (Indeed, exponential moments always decrease under conditional expectation by Jensen's inequality.) Using again the tail-decay definition of sub-Gaussianity completes the proof of \eqref{eq:BChernoff-2}.

The extension to $\mu_q$ poses no additional challenge because we may apply Lemma~\ref{lem:freqguarantee} to $\mu_q=\sum_i q_i\mu_i$, and the estimator $\theta_q=\sum_i q_i\theta_i$ obeys exactly the same Chernoff bound.
\end{proof}

\subsection{Stochastic and MLR Domination}

We use a characterization of multivariate (first-order) stochastic domination \citep{fill2001stochastic}.

\begin{lemma}\label{lem:mvstochdom}
Given two probability distributions $\nu,\nu'$ on $\mathbb R^n$, the following are equivalent:
\begin{OneLiners}
    \item[1.] For any coordinate-wise increasing function $f$ we have $\mathbb E^{\nu}[f]\geq \mathbb E^{\nu'}[f]$.
    \item[2.] There exists a distribution over pairs $(X,X')\in \mathbb R^n\times \mathbb R^n$ such that $X\geq X'$ coordinate-wise almost surely, and $X\sim \nu,X'\sim \nu'$.
\end{OneLiners}

In both cases we will say $\nu$ stochastically dominates $\nu'$. We may also say that $\nu$ is stochastically larger than $\nu'$, call a change from $\nu$ to $\nu'$ a stochastic decrease, etc.

\end{lemma}

In the $1$-dimensional case, there is a \emph{canonical monotone coupling} which is easy to compute when the CDF functions of $\nu,\nu'$ are known. It can be realized explicitly as the random pair $(X,X')=(F_1^{-1}(u),F_2^{-1}(u))$ where $F_1$ is the CDF for $\nu$, $F_2$ is the CDF for $\nu'$, and $u\in [0,1]$ is uniformly random. In the case that $\nu,\nu'$ have finite support, it is easy to sample $X'$ conditionally on $X$ from this coupling by taking
\begin{align}\label{eq:app:dom-smth}
u\in [\; \lim_{\varepsilon\downarrow 0}F_1(X-\varepsilon),F_1(X) \;)
\end{align}
uniformly at random and setting $X'=F_2^{-1}(u)$. This is because \eqref{eq:app:dom-smth}
holds if and only if $F^{-1}(u)=X$ (when we take CDFs and inverse-CDFs to be right-continuous). When $\nu$ stochastically dominates $\nu'$, this coupling satisfies $\Pr[X\geq X']=1$. For general random variables, these couplings always have the property that $X,X'$ are synchronized, i.e. if $(X_1,X_1'),(X_2,X_2')$ are two samples from the coupled distribution, then $X_1\geq X_2$ implies $X_1'\geq X_2'$ with probability $1$. See e.g. \cite{mourrat2020free} for more discusion on monotone couplings.

We give several related lemmas in the $1$-dimensional case. These results will allow us to construct couplings between various posterior mean distributions, for example in the construction of $\FakePi$ in Algorithm~\ref{alg:pyramidtight}.

\begin{definition}

Given mutually absolutely continuous probability distributions $\nu,\nu'$ on $\mathbb R^1$, we say $\nu$ MLR-dominates $\nu'$ if the Radon-Nikodym derivative function $f$ with $\nu(dx)=g(x)d\nu'(x)$ is increasing.

\end{definition}

\begin{lemma}\label{lem:mlr}

If $\nu$ MLR-dominates $\nu'$ then $\nu$ also stochastically dominates $\nu'$.

\end{lemma}

\begin{proof}

Let $f$ be an increasing function and $g=\frac{d\nu}{d\nu'}$ the Radon-Nikodym derivative. We have:

\[\frac{\int_{\mathbb R}f(x)d\nu(x)}{\int_{\mathbb R}f(x)d\nu'(x)}=\frac{\int_{\mathbb R}f(x)g(x)d\nu'(x)}{\int_{\mathbb R}f(x)d\nu'(x)\cdot \int_{\mathbb R}g(x)d\nu'(x)} .\]

We used the fact $\int_{\mathbb R}g(x)d\nu'(x)$ since $\nu,\nu'$ are probability measures. The result follows by FKG.
\end{proof}

We now show some comparison results on posterior mean distributions in $1$ dimension. In the below, $\mu,\mu'\in [0,1]$ always denote the mean reward of bandit arms $a,a'$ which outputs Bernoulli rewards. 

\begin{lemma}\label{lem:mlrdom}
Suppose $\mu\sim \mP$ and $\mu'\sim\mP'$ where distribution $\mP$ MLR-dominates distribution $\mP'$. Then if we observe $1$ sample from each arm and the reward of arm $a$ is at least that of arm $a'$, the posterior distribution of $\mu$ continues to MLR-dominate that of $\mu'$.
\end{lemma}

\begin{proof}
Receiving a reward corresponds to changing $\mP(dx)\to cx\mP(dx)$ for some appropriate constant $c$. Receiving no rewards corresponds to $\mP(x)\to c'(1-x)\mP(dx)$. If $a$ receives at least the reward of $a'$, the MLR property is preserved so the result follows.
\end{proof}

\begin{corollary}\label{cor:1sampledom}

Let $\mu\sim\mP$, and let $\mP_0,\mP_1$ be the posterior distributions for $\mu$ after observing $0,1$ respectively reward from $1$ sample. Then $\mP$ MLR-dominates $\mP$ which MLR-dominates $\mP_0$.

\end{corollary}

\begin{proof}
The result follows from first principles, since observing a reward again corresponds to changing $\mP(dx)\to cx\mP(dx)$ for some appropriate constant $c$. Alternatively, since $\mP$ is a mixture of $\mP_1,\mP_0$ we can reason that $\mP_1$ stochastically dominates $\mP_0$ using Lemma~\ref{lem:mlrdom} and conclude.
\end{proof}

\begin{lemma}\label{lem:mlrposteriordom}
Suppose $\mu\sim \mP$ and $\mu'\sim\mP'$ where $\mP$ MLR-dominates $\mP'$. Let $\mP_N,\mP_N'$ be the distributions of their posterior means after observing $N$ samples each arm. Then $\mP_N$ stochastically dominates $\mP_N'.$
\end{lemma}

\begin{proof}
We reveal the reward of the samples for each arm simultaneously in $N$ steps. We claim there exists a Markovian coupling between the two reward processes so that the posterior distribution of $\mu$ always MLR-dominates that of $\mu'$. This holds by induction: as long as MLR-dominance holds at time $t$, the posterior mean of $\mu$ is higher than that of $\mu'$, so we can couple the rewards so $\mu$ gets at least as much reward. Applying Lemma~\ref{lem:mlrdom}, this preserves the MLR property. At the end, the final posterior of $\mu$ MLR-dominates that of $\mu'$. Therefore we have coupled the rewards for $\mu,\mu'$ so that $\mu$ always has higher posterior mean, implying the claimed stochastic domination.
\end{proof}

\begin{corollary}\label{cor:betapostdom}
Let $\mu'\sim \term{Beta}(a,b)$ with $1\leq a,b\leq M$, and $\mu\sim \term{Beta}(1,M)$. Then for any $N$, the distribution of the posterior mean of $\mu$ after $N$ samples stochastically dominates that of $\mu'$.
\end{corollary}

\begin{proof}
The MLR-domination of $\mu$ over $\mu'$ implies the result.
\end{proof}







\begin{lemma}\label{lem:zerosdom}
Suppose $\mu\sim \mP$. Let $\mathcal L_N$ be the distribution for the posterior mean of $\mu$ after $N$ samples, and $\mathcal L_{N,N_0}$ the distribution for the posterior mean of $\mu$ after $N$ samples, conditioned on the first $N_0$ samples having $0$ reward. Then $\mathcal L_N$ stochastically dominates $\mathcal L_{N,N_0}$.
\end{lemma}

\begin{proof}
If $N_0\geq N$ the result is clear so we assume $N_0\leq N$. Suppose $\mu'$ is an independent copy of $\mu$ corresponding to an arm which observed $N_0$ samples with no reward at the start. After observing the first $N_0$ samples of $\mu$, we obtain a posterior distribution for $\mu$ which MLR-dominates that of $\mu'$. Applying from this point Lemma~\ref{lem:mlrposteriordom}, we have established the result conditionally on any fixed values for the first $N_0$ samples for $\mu$. Averaging over the random choice of first $N_0$ samples now implies the result.
\end{proof}

\subsection{Tails of sub-Gaussian distributions}

\begin{lemma}
\label{lem:subgaussiantail}
If random variable $X$ is $O(1)$-sub-Gaussian and event $E$ has probability $\Pr[E]\leq p$, then $\mathbb E[|X\cdot 1_E|\leq O(p\sqrt{\log(p^{-1})})]$.
\end{lemma}

\begin{proof}

Let $Y:=X\cdot 1_E$. Then we estimate:

\[\Pr[|Y|\geq t]\leq \min(p,\Pr[|X|\geq t]) \leq \min(p,e^{-\Omega(t^2)}).\]
Consequently,
\[\E[|Y|]\leq \int_{t=0}^{\infty} \Pr[|Y|\geq t]dt \leq \int_{t=0}^{\infty} \min(p,\Pr[|X|\geq t])dt. \]
The integrand is $t$ for $t=O(\sqrt{\log(p^{-1})})$ and $e^{-\Omega(t^2)}$ elsewhere. Combining gives the claimed bound.
\end{proof}

\section{Initial Sampling: proofs for Section~\ref{sec:sampling}}

\subsection{Existence of a Suitable Policy: proof of Lemma~\ref{lm:policy-existence}}




We first formally define the $j$-recommendation game. We proceed more generally than in the main body, defining it relative to any $K$-tuple $(N_1,\dots,N_K)$. We recall the definition of the static $\sigma$-algebra $\mathcal G_{N_1,\dots,N_K}$ which is generated by $N_i$ samples of each arm $i$. When considering arm $j$, we will always have $N_k=0$ for all $k>j$. If $N_i=N$ for all $i$ we recover the $(j,N)$-recommendation game as defined in Section~\ref{sec:sampling}. The definition of $(j,N)$-informed generalizes readily to $(j,\mathcal G)$-informed.

\begin{definition}

The \emph{$j$-recommendation game} is a two-player zero sum game played between a \emph{planner} and an \emph{agent}. The players share a common independent prior over the true mean rewards $(\mu_i)_{i\leq j}$, and the planner gains access to the static $\sigma$-algebra $\mathcal G=\mathcal G_{(N_1,\dots,N_j)}$. The planner then either recommends arm $j$ or does nothing. The agent picks a mixed strategy on $[j-1]$, resulting in the choice of an arm $i\leq j-1$. Suppose the arms chosen are $a_i$ for the planner, $a_k$ for the agent (where $k=j$ or $k$ is undefined). The payoff for the planner is $(\mu_k-\mu_i)\cdot 1_{k=j}.$

\end{definition}

We remark we can without loss of generality view all $j$-recommendation game strategies as depending only on the posterior means of each arm condition on $\mathcal G$. In the below, we set $\widetilde\mu_i=\mathbb E[\mu_i|\mathcal G]$ for the relevant static $\sigma$-algebra $\mathcal G$.  Given a planner strategy in the $j$-recommendation game, we naturally obtain a corresponding  $(j,\mathcal G)$-informed policy for our original problem in which we recommend arm $j$ when the planner would, and recommend the $\mathcal G$-conditional-expectation-maximizing arm otherwise. The key point of this game is as follows.

\begin{lemma}\label{lem:gameBIC}

A strategy for the planner in the $j$-recommendation game corresponds to a $(j,\padding)$-padded  BIC policy if and only if it has minimax value at least $\padding$.


\end{lemma}

\begin{proof}
Because this policy only plays an arm $i\neq j$ when $\mathbb E[\mu_i|\mathcal G]$ is maximal among all arms, it suffices to show that recommending arm $k$ is $(j,\padding)$-padded BIC. This follows directly from the definition.
\end{proof}

\begin{lemma}\label{lem:gamevalue}

The minimax value of the $j$-recommendation game with $\sigma$ algebra $\mathcal G$ is equal to

\[\inf_{q\in\Delta_{j-1}}\E\left[\left(\E\left[\mu_j-\mu_q\mid \mathcal G\right]\right)_+\right]=\inf_{q\in\Delta_{j-1}}\E\left[\left(\widetilde\mu_j-\widetilde\mu_q\right)_+\right].\]

\end{lemma}

\begin{proof}
This expression is the value of the best response when the agent plays mixed strategy $q$. By the minimax theorem, the infimum over all $q\in\Delta_{j-1}$ is exactly the value of the game.
\end{proof}

\begin{lemma}\cite[Corollary 5.5]{apt2011primer}\label{lem:goodnash}
In any finite two-player zero-sum game there is a Nash equilibrium of mixed strategies which are not weakly dominated by any other mixed strategy.
\end{lemma}

The $j$-recommendation game can be viewed as a finite game because there are only finitely many datasets due to the Bernoulli reward assumption. Hence the above lemma applies.

\begin{definition}

We say a policy $\pi$ (for the original bandit problem) is $j$-correlated if $\mathbb P[\pi=j|(\widetilde\mu_i)_{i\in [K]}]$ is increasing in $\widetilde\mu_j$.


\end{definition}

\begin{lemma}\label{lem:weakdomjcor}

If a strategy $S$ for the $j$-recommendation game is not weakly dominated by any other policy, then the resulting policy for the original bandit problem is $j$-correlated.

\end{lemma}

\begin{proof}
If $\mathbb P[S=j|(\widetilde\mu_i)_{i\in [K]}]$ is not increasing in $\E[\mu_j|\mathcal G]$ then there exist two $j$-tuples \[(e_1,\dots,e_j),(e_1,\dots, e_{j-1},e_j')\] of values for $(\widetilde\mu_i)_{i\leq j}$ which agree except on $e_j>e_j'$ and yet

\[\Pr[S =j|(\widetilde\mu_i)_{i\in [K]}=(e_1,\dots,e_j)]< \Pr[S' \text{ recommends } a_j |(\widetilde\mu_i)_{i\in [K]}=(e_1,\dots,e_j')].\]

In this case we may decrease the probability to play $a_j$ on former event that $(\widetilde\mu_i)_{i\in [K]}=(e_1,\dots,e_j)$, and proportionally increase this probability given $(\widetilde\mu_i)_{i\in [K]}=(e_1,\dots,e_j')$, so that the total probability to play arm $j$ is preserved. This strategy weakly dominates $S'$, yielding a contradiction.
\end{proof}

\begin{lemma}\label{lem:jcorbic}

If a strategy for the $j$-recommendation game is $j$-correlated and $(j,0)$-padded, then the resulting $(j,\mathcal G)$-informed policy is BIC.

\end{lemma}

\begin{proof}

 The BIC property against arms $i<j$ is equivalent to the $(j,0)$-padded property. By the first property of $j$-correlation it follows that the event of playing arm $j$ is increasing in $\mathbb E[\mu_j|\mathcal G]$, hence the probability to play arm $j$ is stochastically increasing in $\mu_j$. Therefore FKG implies that
\[\mathbb P[\mu_j|A_t=j]\geq \mu_j^0\geq \mu_i^0\]
for any $i>j$. It remains to show the BIC property for playing arms other than arm $j$, and this follows from the fact that we simply exploit conditionally on $\mathcal G$ whenever this happens.
\end{proof}

\begin{lemma}\label{lem:kcorrelated}

Given a static $\sigma$-algebra $\mathcal G$, suppose there exists a $\mathcal G$-measurable $(j,\padding)$-padded BIC strategy $\hat\pi_j$ for $\padding\geq 0$. Then there exists a $\mathcal G$-measurable $(j,\padding)$-suitable BIC strategy $\pi_j$.

\end{lemma}

\begin{proof}

The assumption implies there exists a $j$-recommendation game strategy $S$, and any such strategy yields a $(j,\padding)$-padded BIC policy. By Lemma~\ref{lem:goodnash} we may assume $S$ is not weakly dominated. In such case, by Lemmas~\ref{lem:weakdomjcor} and \ref{lem:jcorbic} the resulting bandit strategy $\pi_j$ is BIC as well.
\end{proof}









Combining the previous lemmas, we obtain the main guarantee for $(j,\padding)$-suitable strategies.

\begin{restatable}{lemma}{BICmargin}
\label{lem:BICmargin}
Fix $j\in [K]$ and $\padding>0$. Suppose there exists a static $\sigma$-algebra $\mathcal G=\mathcal G_{N_1,N_2,\dots,N_k,0,\dots,0}$ which is independent of
    $(\mu_{j+1}  \LDOTS \mu_K)$
and satisfies
\[\E\left[\left(\E\left[\mu_j-\mu_q\mid \mathcal G\right]\right)_+\right]
\geq \padding \quad\text{ for all }q\in\Delta_{j-1}.\]
Then there exists a $\mathcal G$-measurable $(j,\padding)$-suitable strategy for the planner.

\end{restatable}

\begin{proof}

Follows from combining Lemmas~\ref{lem:gameBIC}, \ref{lem:gamevalue}, \ref{lem:kcorrelated}.
\end{proof}

The next lemma upper bounds the number of samples we must include in $\mathcal G$ to ensure $\padding\geq\Omega(\padG)$.

\begin{lemma}
\label{lem:chernoff2}
Fix arm $j$ and time $T_0$, and a sufficiently large universal constant $\padC\geq 5$. Suppose for some $q\in\Delta_{j-1}$ and
    $\eps_{\post},\delta_{\post}>0$,
\[ \Pr\left[\left(\mu_j-\mu_q\right)_+ \geq \eps_{\post}\right]\geq \delta_{\post}.\]
Let \ALG be a BIC algorithm which by time $T_0$ collects at least $N_{\post}$ samples of each arm $i\in [j]$ a.s., where
\[N_{\post}\geq
\padC\;\eps_{\post}^{-2}
\left(1+\log\left(\eps_{\post}^{-1}\cdot\delta_{\post}^{-1}\right)
\right)\]
Let $\mathcal G_{N_{\post},k}$ be the $\sigma$-algebra generated by the first $N_{\post}$ samples of each of these $k$ arms. Then
\[\E\left[\left(\E\left[\mu_j-\mu_q \mid \mathcal G_{N_{\post},k}\right]\right)_+\right]\geq \frac{\E[(\mu_j-\mu_q)_+]-\delta_{\post}-\eps_{\post}}{2}.\]

\end{lemma}

\begin{proof}
Applying Lemma~\ref{lem:TSchernoff} and using that $\padC\geq 5$ is sufficiently large, for any $r\geq 1$ we have
\[\Pr\left[|\E[\mu_q\mid\mathcal \mathcal G_{N_{\post},j}]-\mu_q|\geq \frac{r\eps_{\post}}{10}\right]\leq \padC e^{-100\padC^2 r^2\log(\eps_{\post}\delta_{\post})}\leq \padC (\eps_{\post}\delta_{\post})^{100\padC r}.\]
This easily implies by integration that
\[\E\left[\left(|\E[\mu_q\mid \mathcal G_{N_{\post},j}]-\mu_q|- \frac{\eps_{\post}}{10}\right)_+\right]\leq
\frac{\eps_{\post}\delta_{\post}}{10}. \]

Similarly,
\[\E\left[\left(|\E[\mu_j\mid \mathcal G_{N_{\post},j}]-\mu_j|- \frac{\eps_{\post}}{10}\right)_+\right]\leq \frac{\eps_{\post}\delta_{\post}}{10}. \]

However we have
\begin{align*}\left(\E\left[\mu_j-\mu_q|\mathcal G_{N_{\post},j}\right]\right)_+ &\geq \left(\mu_j-\mu_q - \frac{\eps_{\post}}{5}\right)\cdot 1_{\mu_j- \mu_q\geq\eps_{\post}}- \left(|\E[\mu_j|\mathcal G_{N_{\post},j}]-\mu_j|- \frac{\eps_{\post}}{10}\right)_+ \\&-\left(|\E[\mu_q|\mathcal G_{N_{\post},j}]-\mu_q|- \frac{\eps_{\post}}{10}\right)_+. \end{align*}

Note that $\left(\mu_j-\mu_q - \frac{\eps_{\post}}{5}\right)\cdot 1_{\mu_j- \mu_q\geq\eps_{\post}}\geq \frac{(\mu_j-\mu_q)_+}{2}\cdot 1_{\mu_j- \mu_q\geq\eps_{\post}}.$ Substituting and taking expectations, the conclusion of the lemma follows.
\end{proof}

We now obtain Lemma~\ref{lm:policy-existence}, guaranteeing the existence of the policies $\pi_j$ needed in Algorithm~\ref{alg:pyramidtight}.

\begin{proof}[Proof of Lemma~\ref{lm:policy-existence}]
Letting $\eps_{\post}=\delta_{\post}=\padG/3$ in Lemma~\ref{lem:chernoff2}, the result follows immediately. To see that the assumption of Lemma~\ref{lem:chernoff2} holds note that $(\mu_j-\mu_q)_+\in [0,1]$ almost surely. Therefore \[\Pr[(\mu_j-\mu_q)_+\geq \padG/3]\geq \padG/3\] as long as $\mathbb E[(\mu_j-\mu_q)_+]\geq \padG.$
\end{proof}

\subsection{A suitable policy for bootstrapping: proof of Lemma~\ref{lem:transform}}

We consider $\FakePi$, a transformed version of policy $\pi_j$ which is, essentially,  $(j,\padding,N)$-suitable conditioned on $\ZEROS_{j,N_0}$.


Existence of such a policy is straightforward, as the conditioning on $\ZEROS_{j,N_0}$ stochastically decreases the values $(\widetilde{\mu}_i)_{i<j}$, hence increasing the minimax value of the $j$-recommendation game. However we give a reduction to $\pi_j$, so that to make Algorithm~\ref{alg:pyramidtight} efficient it essentially suffices to compute policies $\pi_j$.

We use the following fake data technique. Let $\widetilde{\mu}_i$ be the posterior mean $\mathbb E[\mu_i|\mathcal G_{j,N}]$, and $\mathcal L_i$ be the distribution of $\widetilde{\mu}_i$. Let $\mathcal L_i'$ be the distribution of $\widetilde{\mu}_i$ conditioned on $\ZEROS_{j,N_0}$. It follows from Lemma~\ref{lem:zerosdom} that $\mathcal L_i$ is stochastically dominated by $\mathcal L_i'$. Hence by Lemma~\ref{lem:mvstochdom} there exists a coupling $(\widetilde{\mu}_i,\hat{\mu}_i)$ with $\widetilde{\mu}_i\leq\hat{\mu}_i$ almost surely, with $\widetilde{\mu}_i\sim \mathcal L_i,\hat\mu_i\sim\mathcal L_i'.$

We now define $\FakePi.$ We view $\pi_j=\pi_j(\widetilde\mu_1,\dots\widetilde\mu_j)$ as a policy of posterior mean rewards; this makes no difference because for fixed prior on some $\mu_i$ and a fixed number of samples, specifying the posterior mean $\widetilde\mu_i$ is equivalent to specifying the empirical average reward. We set $\FakePi(\widetilde\mu_1,\dots\widetilde\mu_j)=j$ when $\pi_j(\hat\mu_1,\dots,\hat\mu_{j-1},\widetilde\mu_j)=j$. That is, we use the couplings $\mathcal X_i$ to generate overly optimistic posterior mean rewards for all arms except arm $j$, and then apply $\pi_j$. When $\pi_j(\hat\mu_1,\dots,\hat\mu_{j-1},\widetilde\mu_j)\neq j$ we simply exploit based on the first $N$ samples of each arm $1,\dots,j$ (which may result in sampling arm $j$ anyway).

We proceed to show $\FakePi$ is $(j,\padding)$-suitable. It suffices to show arm $j$ recommendations are BIC since otherwise $\FakePi$ exploits. Fix $i<j$. Below we let $\pi_j$ denote the value $\pi_j(\hat\mu_1,\dots,\hat\mu_{j-1},\widetilde\mu_j)$. The key point is that the distribution of $(\hat\mu_1,\dots,\hat\mu_{j-1},\widetilde\mu_j)$ conditioned on $\ZEROS_{j,N_0}$ is, by construction, the unconditioned distribution for $(\widetilde\mu_1,\dots\widetilde\mu_j)$. Therefore, by the $(j,\padding)$-padded BIC property of $\pi_j$:
\[\mathbb E[(\widetilde\mu_j-\hat\mu_i)\cdot 1_{\pi_j=j}|\ZEROS_{j,N_0} ]\geq \padding.\]

As $\widetilde\mu_i\geq \hat\mu_i$ almost surely we obtain
\[\mathbb E[(\widetilde\mu_j-\widetilde\mu_i)\cdot 1_{\pi_j=j}|\ZEROS_{j,N_0} ]\geq \padding.\]
 for each $i<j$. On the event $\pi_j\neq j$ we always exploit, implying that
 \[\mathbb E[(\widetilde\mu_j-\widetilde\mu_i)\cdot 1_{\pi_j\neq j}|\ZEROS_{j,N_0} ]\geq 0.\]

 Adding, we obtain the padded BIC property conditional on $\ZEROS_{j,N_0}$. Compared against arms $i>j$, the conditional BIC property for $\FakePi$ follows directly from the unconditioned property of $\pi_j$. Indeed, the couplings restore the distribution of the first $j-1$ arm samples back to normal, and do not affect the remaining $N+1-j$ arms. This concludes the proof.


\subsection{BIC property for bootstrapping: proof of Lemma~\ref{lm:main-BIC0}}

\mainBIC*

\begin{proof}[Proof of Lemma~\ref{lm:main-BIC0}]
The proof follows the same strategy as that of Lemma~\ref{lm:main-BIC}. Line~\ref{line:initialexploit} is BIC since it is just exploitation. For the next phase, recommending any arm $i\neq j$ is BIC since it can only happen via exploitation, or via the BIC policy $\FakePi$. Hence we fix an arm $i\neq j$ and show that arm $j$ is BIC against arm $i$ in this phase. We define
$\IndPadded+\IndExploit+\IndExplore=1$ and $\Lambda_{i,j}$ as in the proof of Lemma~\ref{lm:main-BIC}. Note that we have probability $\Pr[\ZEROS_{j,N_0}]\cdot (1-p_j)$ to reach line~\ref{line:padded1}. Applying the guarantee of Lemma~\ref{lem:transform} conditional on this event, we obtain:
\begin{align*}\E[\Lambda_{i,j}\cdot \IndPadded]&=\E[(\mu_j-\mu_i)\cdot 1_{\FakePi=j}\cdot \IndPadded]\\
&=\E[(\widetilde \mu_j-\widetilde \mu_i)\cdot 1_{\FakePi=j}\cdot \IndPadded]\\
&\geq \E[(\widetilde \mu_j-\hat \mu_i)\cdot 1_{\FakePi=j}\cdot \IndPadded]\\
& \geq \Pr[\ZEROS_{j,N_0}]\cdot \padding (1-p_j)\\
 &= p_j.\end{align*}

The second equality above follows from the fact that both $1_{\FakePi=j}$ and $\IndPadded$ only depend on the first $N$ samples of the first $j$ arms and independent external randomness.

Moreover using a worst-case bound on exploration, and that exploitation is always BIC, we have:

\[\E[\Lambda_{i,j}\cdot \IndExplore]\geq -\E[\IndExplore]=-p_j\quad\text{ and }\quad  \E[\Lambda_{i,j}\cdot \IndExploit]\geq 0.\]

Adding, we obtain Equation~\ref{eq:BICcondition} for $i<j$. For $i>j$, similar reasoning again leads to the same three inequalities with $p_j$ replaced by $0$, i.e. $\E[\Lambda_{i,j}\cdot 1_E]\geq 0$ for $1_E\in \{\IndPadded,\IndExploit,\IndExplore\}.$ Altogether we have shown that playing arm $j$ is BIC against any arm $i$, establishing the BIC property for the bootstrapping phases.
\end{proof}

\section{Sample Complexity for Arbitrary Priors (for Section~\ref{sec:sample})}











The purpose of this Appendix is to estimate the parameters $\padG,N_{\TS},\padN,\bootN,\log\left(\Pr[\ZEROS_{j,\bootN}]^{-1}\right)$ appearing in the sample complexity upper bound of Theorem~\ref{thm:goodalgo}. In Theorem~\ref{thm:tscomplexity} we achieve upper bounds based on the lower bounds $K,\bootN,\MainLB$ for $\Tbic{1}$ as well as a standard deviation lower bound $\sigma$. Then in Lemma~\ref{lem:easyexploreparameters} we show upper bounds depending on $\delta$ for $\delta$-easy and $\delta$-non-dominant problem instances.

In Theorem~\ref{thm:tscomplexity} we use the following definition, a type of anti-concentration assumption on the priors. It is implied (up to constant factors) by all $K$ arms having standard deviation at least $\sigma$. However it also allows for a single highly concentrated arm, as long as this arm has some chance to be away from $1$.

\begin{definition} \label{defn:epsndg} The priors are called \emph{$\sigma$-non-degenerate} if
    (i)  $\term{StdDev}(\mu_i) \geq \sigma$ for all but at most one arm $i$,
and (ii) $\Pr[\mu_i\in [0,1-\sigma]] \geq e^{-1/\sigma}$
for all arms $i$.
\end{definition}

The next result is a more refined statement of Corollary~\ref{thm:matching}, incorporating Corollary~\ref{eq:thm:matching-improved} as well via Definition~\ref{defn:epsndg}.

\begin{theorem}
\label{thm:tscomplexity}
Assume the priors are $\sigma$-non-degenerate. Then Algorithm~\ref{alg:pyramidtight} achieves $N_{\TS}$ samples of each arm within this many rounds:
\[\Tbic{N_{\TS}} =\tilde O\left(\frac{K^{9/2}\MainLB^3\Nexplore}{\sigma^4}+\frac{K^{5/2} \MainLB\Nexplore^2}{\sigma^{2}}\right).\]
Moreover Algorithm~\ref{alg:pyramidtightfast} (defined later), achieves $N_{\TS}$ samples of each arm within this many rounds:
\[\Tbic{N_{\TS}}= \tilde O\left(\frac{K^{7/2}\MainLB^3\Nexplore}{\sigma^4}\right).\]
\end{theorem}

\begin{proof}
We use the guarantees from Theorem~\ref{thm:goodalgo} and
\eqref{eq:alg2-guarantee},
and we upper-bound each of the relevant parameters in terms of $\MainLB,\bootN,K,\sigma$. Recall from Theorem~\ref{thm:TSsamples} that
$N_{\TS} = C_{\TS}\,\eps_{\TS}^{-2}\;\log \delta_{\TS}^{-1}$,
where $\delta_{\TS}\geq \eps_{\TS}^K$ by FKG (see the proof of Lemma~\ref{cor:TSsamples}) so that $N_{\TS}=O\left(\frac{K\log(1/\eps_{\TS})}{\eps_{\TS}^2}\right).$

From the condition that $\Pr[\mu_i\in [0,1-\sigma]]\geq e^{-K\sigma^{-1}}$, it is not hard to obtain
    \[\Pr[\ZEROS_{j,\bootN}]\geq e^{-O(K\Nexplore\sigma^{-1})},\]
which implies \[\log\left(\frac{1}{\Pr[\ZEROS_{j,\bootN}]}\right)=O\left(K\Nexplore\sigma^{-1}\right).\]

Define \[\widehat L:= \sup_{q\in\Delta,j\in[K]:q_j=0} \frac{1}{\E[(\mu_j-\mu_q)_+]}.\] We have $\frac{1}{\padG}\leq\widehat L$ as $\widehat L$ minimizes the same objective over a large set. We may take
    $\eps_{\TS}=\frac{1}{10\widehat L}$,
so that $N_{\TS}=\tilde O(K\hat{L}^2).$ Recall that $N_{\post}=O\left(\frac{\log(\padG^{-1})}{\padG^2}\right)=\tilde O(\hat{L}^2)$. Combining we obtain the bounds
\[ \tilde O\left( K^2\Nexplore \widehat L(K\widehat L^2+\Nexplore)\sigma^{-1} \right)
    \quad\text{and}\quad
 \tilde O\left(K^2\widehat L^3 \Nexplore\sigma^{-1}\right)\]
on the number of rounds needed by Algorithms~\ref{alg:pyramidtight} and~\ref{alg:pyramidtightfast} respectively. Next we estimate:
\begin{align*}\widehat L&=\sup_{q\in\Delta_K,j\in[K]:q_j=0}\left(\frac{1}{\E[(\mu_j-\mu_q)_+]}\right) \\
&\leq\sup_{q\in\Delta_K,j\in[K]:q_j=0}\frac{\E[(\mu_j-\mu_q)_+ +(\mu_j-\mu_q)_-]}{\E[(\mu_j-\mu_q)_+]}\cdot \sup_{q\in\Delta_K,j\in[K]:q_j=0}\frac{1}{\E[(\mu_j-\mu_q)_+ +(\mu_j-\mu_q)_-]}\\
&\leq\sup_{q\in\Delta_K,j\in[K]:q_j=0}\frac{\E[(\mu_j-\mu_q)_+ +(\mu_j-\mu_q)_-]}{\E[(\mu_j-\mu_q)_+]}\cdot \sup_{q\in\Delta_K,j\in[K]:q_j=0}\frac{1}{\E[|\mu_j-\mu_q|]}\\
&\leq O\left(\frac{\MainLB\sqrt{K}}{\sigma}\right).\end{align*}

In the last step, we use the assumption that at most one $\mu_i$ has standard deviation less than $\sigma$ to bound the last supremum by $O\left(\frac{\sqrt{K}}{\sigma}\right)$. This bound follows because either $\mu_k$ has standard deviation at least $\sigma$, or $\mu_q$ has standard deviation at least $\frac{\sigma}{\sqrt{K}}$. The claimed bounds follow.
\end{proof}


\begin{lemma}\label{lem:closetoEV}

If $X\in [0,1]$ almost surely then $\mathbb P[X\leq \mathbb E[X]+\varepsilon]\geq \frac{\varepsilon}{1+\varepsilon}.$

\end{lemma}

\begin{proof}
Let
    $p=\mathbb P[X\leq \mathbb E[X]+\varepsilon]$.
Then
    $\mathbb E[X-\mathbb E[X]]\geq p\varepsilon-(1-p)$.
However, this is false if
    $p< \frac{\varepsilon}{1+\varepsilon}$.
\end{proof}

In the next lemma we estimate all the parameters appearing in the guarantee of Theorem~\ref{thm:goodalgo} for Algorithm~\ref{alg:pyramidtight} for $\delta$-easy and $\delta$-non-dominant problem instances.

\begin{restatable}{lemma}{easyexploreparameters}
\label{lem:easyexploreparameters}
If $\mC$ is $\delta$-easy and $\delta$-non-dominant, then:
    $N_{\TS}=\tilde O(K\delta^{-2})$ and
    $\padG\geq \delta$ and
    $\padN=\tilde O(\delta^{-2})$  and
    $\bootN=\tilde O(\delta^{-1})$  and
    $\log(\Pexplore^{-1})=\tilde O(K\delta^{-1})$.
\end{restatable}

\begin{proof}

We begin with the first assertion. Since $\mathbb E\sbr{ (\mu - \Phi^{\sup}_\mC)_+}>\delta$ for any $\mu_j\sim\mP_j\in\mC$ we see that $\mathbb P[\mu_j\geq \Phi^{\sup}+\frac{\delta}{2}]\geq \frac{\delta}{2}$. On the other hand by Lemma~\ref{lem:closetoEV} with $\varepsilon=\frac{\delta}{3}$ we have for each $\mu_i\sin \mP_i$ that
\[\mathbb P[\mu_i\leq \Phi^{\sup}+\frac{\delta}{3}]\geq \mathbb P[\mu_i\leq \mathbb E[\mu_i]+\frac{\delta}{3}]\geq \Omega(\delta).\]

We conclude that $\mathbb P[A^*=\mu_j]\geq \Omega(\delta)^K$ for any $j$ and therefore that $\delta^{-1}_{\TS}=O(K\log(\delta^{-1})).$ On the other hand by definition:
\[\mathbb E[(\mu_j-\mu_i)_+]\geq \mathbb E[(\mu_j-\Phi^{\sup}_{\mC})_+]\geq \delta \]
which implies $\varepsilon_{\TS}\geq \delta$ and hence $N_{\TS}=\tilde O(K\delta^{-2}).$

To see that $\padG\geq \delta$ we observe that for any $i\in [K],q\in\Delta_{i-1}$ we have by convexity:
\begin{align*}\mathbb E[(\mu_i-\mu_q)_+]&\geq \mathbb E[(\mu_i-\mu_q^0)_+]\\
&\geq \mathbb E[(\mu_i-\Phi^{\sup}_{\mC})_+]\\
&\geq \delta.\end{align*}

The fact that $\padG\geq\delta$ immediately implies $N_{\post}=\tilde O(\delta^{-2})$.

Finally, for $\bootN$ and $\Pexplore$ we rely also on $\delta$-explorability. To estimate $\bootN$ we observe that every $\delta^{-1}$ zero-reward samples of arm $i$ give a constant factor likelihood ratio advantage of any $\mu_i\leq \mu_j^0-\delta/2$ over any $\mu_i\geq \mu_j^0$. Since $\mathbb P[\mu_i\leq \mu_j^0-\delta/2]\geq \frac{\delta}{2}$ it is easy to see that $\bootN=\tilde O(\delta)$ from this. Since each zero-reward event has probability at least $1-\mu_i^0\geq \Omega(\delta)$ given the previous ones,
$\log(\Pexplore^{-1})=\tilde O(K/\delta)$.
\end{proof}

\subsection{Necessity of the Non-Degeneracy Assumption}
\label{app:non-degeneracy}

We previously showed that under a $\sigma$-non-degeneracy assumption, Thompson sampling can be made BIC after an amount of time polynomial in $\sigma^{-1}$ and the time required to sample every action. It is natural to wonder if the dependence on an additional parameter $\sigma$ is necessary, or if these two quantities are always polynomially related. Here we show that dependence on an additional parameter is unavoidable. In particular we give an $\varepsilon$-non-degenerate problem instance in which both arms can be sampled in $2$ rounds, but $\Omega(\varepsilon^{-1})$ rounds are needed to make Thompson sampling BIC. In particular, the polynomial dependence on $\sigma$ in Corollary~\ref{thm:matching} is necessary in this example.

\begin{proposition}
\label{prop:degenexample}
Consider an initial prior on two arms where $\mu_1=\frac{1}{2}\pm\eps$ where the sign is chosen uniformly at random, and $\mu_2=\frac{1}{2}-\frac{\eps^2}{10}$ almost surely. Then it is possible to sample both arms in time $O(1)$. However, if there is a BIC algorithm which almost surely uses Thompson sampling on round $t$, then $t\geq \Omega(\eps^{-1})$.
\end{proposition}

\begin{proof}

First we explain how to sample both arms in time $O(1)$. We first sample arm $1$. Note that if the observed reward $r_1$ is $0$, then the posterior mean for $\mu_1$ is $\frac{1}{2}-2\eps^2$ while if the reward is $1$ then the posterior mean is $\frac{1}{2}+2\eps^2$. Therefore, if we in the next round sample $\mu_2$ with probability $1$ when $r_1=0$, and with probability $\frac{1}{2}$ when $r_1=1$, this is also BIC. Doing this twice with appropriate coupling allows us to sample both arms almost surely in $3$ rounds.

Now we prove the lower bound for Thompson sampling. We recall that Thompson sampling at time $t$ is BIC if and only if \[\E^{t}[\mu_i-\mu_j]\cdot {\Pr}^{t}[A_t=i]\geq 0.\] The proof of Lemma~\ref{thm:TSBIC} shows that the left-hand side of the above equation, which we will call $X_t$, is a submartingale. Now, sampling arm $2$ gives no information, hence no change in $X_t$. Therefore we may define $Y_m$ to be the value of $X_t$ at any time when arm $2$ has been sampled $m$ times. It is not hard to see that $Y_m$ is also a submartingale, and is independent of the choice of which arm to sample. (We may define $Y_m$ for all $m$ by generating infinity many samples of arm $2$ ``in secret" and thinking of sampling arm $2$ as revealing the next sample.) Since $Y$ is a submartingale, we by definition have $\E^t[Y_t]\geq \E^t[X_t]$ almost surely for any bandit algorithm. Therefore $\mathbb E[Y^t]\geq \E[X^t]$ for any algorithm. We conclude that if there exists a BIC algorithm which uses Thompson sampling almost surely at time $t$, then sampling arm $1$ for the first $t-1$ rounds also works. After $t-1$ samples of $a_1$, we gain $O(t\eps^2)$ bits of information on the value of $\mu_1$. Therefore any function of our observations is correlated at most $O(t\eps^2)$ with $1_{\mu_1>\frac{1}{2}}$, and hence $O(t\eps^3)$ correlated with the value $\mu_1$. This implies that $t=\Omega(\eps^{-1})$ is required for TS to be BIC.
\end{proof} 

\section{Sample Complexity for Truncated Gaussians and Beta priors}

In this Appendix we consider concrete problems with truncated Gaussian and Beta priors, determining the sample complexity for Thompson sampling up to polynomial dependence in several situations. Our strategy throughout is to methodically estimate the parameters $\padG,\bootN,\MainLB$, and so on. The next lemma allows us to simplify many of these computations using symmetry.

\begin{lemma}\label{lem:sym}

Suppose $\mu_1,\dots,\mu_{j-1}$ are independent and identically distributed. Then

\[\sup_{q\in\Delta_{j-1}} \frac{\E[\mu_q^0-\mu_j^0]}{\E[(\mu_j-\mu_q)_+]}\]
is achieved at $q=\left(\frac{1}{j-1},\frac{1}{j-1},\dots,\frac{1}{j-1}\right).$

\end{lemma}

\begin{proof}

It suffices to show that $\E[(\mu_j-\mu_q)_+]$ is minimized at the claimed value of $q$, since $\mu_q^0$ is independent of $q\in\Delta_{j-1}$. Indeed $\E[(\mu_j-\mu_q)_+]$ is a convex, symmetric function of $q\in\Delta_{j-1}$ so it must be minimized when all coordinates are equal. The symmetry is clear while convexity holds for any fixed values of $(\mu_i)_{1\leq i\leq j}$, hence in expectation.

\end{proof}

The next lemma allows comparison with stochastically dominating problem instances to estimate $\MainLB,G$.

\begin{lemma}\label{lem:stochdomcompare}

For any $j$ and $q\in \Delta_K$ with $q_j=0$, the expressions

\[\frac{\E[\mu_q^0-\mu_j^0]}{\E[(\mu_j-\mu_q)_+]},\quad \frac{1}{\E[(\mu_j-\mu_q)_+]}\]

are stochastically decreasing in the prior $\mP_j$ for $\mu_j$ and stochastically increasing in all the priors $\mP_i$ for $\mu_i$ when $i\neq j$.

\end{lemma}

\begin{proof}

Based on Lemma~\ref{lem:mvstochdom} it suffices to show the relevant monotonicity of each part of the expressions without expectations in $\mu_j$ and $\mu_i$, which is clear.

\end{proof}

\begin{lemma}\label{lem:gaussstochdom}

For any $ \mu<\mu'$, the distribution $\widetilde N(\mu,\sigma)$ is stochastically smaller than the distribution $\widetilde N(\mu',\sigma)$.

\end{lemma}

\begin{proof}

In fact MLR domination holds, and hence stochastical domination follows. Observe that the densities $f(x), g(x)$ for $\widetilde N(\mu,\sigma)$ and $\widetilde N(\mu',\sigma)$ satisfy
$f(a)/g(a)\leq f(b)/g(b)$
for any $a\leq b$. In fact the ratio is proportional to $e^{(a-b)\cdot (\mu'-\mu)}$.
\end{proof}




\gaussianbound*

\begin{proof}
We assume without loss of generality that $\sigma$ is at most a small constant. First, the distributions are clearly $\Omega(\sigma)$-nondegenerate so we focus on bounding the values $\delta$ and $\MainLB$. The upper bound then follows based on Lemma~\ref{lem:easyexploreparameters}.


We note that the mean of $\widetilde N(\nu_i,\sigma^2)$ is $\nu_i\pm O(\sigma)$ because (as $\sigma=O(1)$) we are conditioning on a probability $\Omega(1)$ event in truncating the Gaussian to $[0,1]$. Similarly any such distribution has mean $\Omega(\sigma)$ and $1-\Omega(\sigma)$.

Let $\mC$ be the set of $K$ priors in the problem, we claim $\mC$ is $\delta$-easy and $\delta$-non-dominating for $\delta^{-1}\leq poly(\sigma^{-1},e^{R^2}).$ Since the conditions are symmetric under $\nu_i\to 1-\nu_i$ we focus on $\delta$-non-dominance. First it is easy to see that $\mu_i^0\geq\Omega(\sigma)$ for each arm $i$. The fact that for any $i,j$ and constants $C,c$ we have
    \[\mu_i^0\geq \max\rbr{ c\sigma,\nu_j-(R+C)\,\sigma}, \]
which implies that
\[\mathbb E[(\mu_i^0-\mu_j)_+]\geq \Omega\left(\sigma \mathbb P[\mu_j\leq \max(c\sigma/2,\nu_j-2(R+C)\sigma)] \right).\]

It is easy to see that the density of $\mu_j$ is at least $e^{-O(R^2)}$ in the interval
\[\max\left(c\sigma/4,\left(\nu_j-2(R+C)-c/4\right)\sigma\right).\]

Hence $\mathbb P[\mu_j\leq \max(c\sigma/2,\nu_j-2(R+C)\sigma)]\geq \Omega(\sigma e^{R^2})$ and so we obtain the claimed lower bound on $\delta$. Using the general $\delta$-dependent upper bounds now yields the result.

We next turn to estimating $\MainLB$ to achieve a lower bound. Suppose $m_i+R\sigma=m_j$ with $R\sigma=\Omega(1)$ of constant order. Then we easily see that
\[\mathbb P[\mu_j\geq \mu_i]\leq e^{-\Omega(R^2)}.\]

Indeed this is clear for non-truncated Gaussians, and conditioning $\mu_i,\mu_j\in [0,1]$ restricts to a constant probability event, hence can only increase probabilities by a constant factor. Therefore
\[\MainLB\geq \frac{\nu_j-\nu_i}{\mathbb P[\mu_j\geq \nu_i]}=e^{\Omega(R^2)}(\nu_j-\nu_i).  \]

Since we may take $i,j$ to maximize $|\nu_i-\nu_j|$ this shows the claimed lower bound (as the assumptions imply $R=\Omega(\sigma^{-2})$.
\end{proof}

\betabound*

\begin{proof}

For convenience we replace the strength condition with the (up to constant factors) equivalent condition that all \term{Beta} parameters are at most $M$. We first estimate $N_{\TS},\Nexplore,\log(\Pexplore^{-1})$ towards establishing the upper bound. We estimate $\padG,\MainLB$ together at the end, which implies the upper bound via Theorem~\ref{thm:pyramidtightfast} and the lower bound directly. Corollary~\ref{cor:tslogk} gives $N_{\TS}=2^{O(M)}$. Computing $\Nexplore$ may be done exactly in the case when one variable is $\term{Beta}(1,M)$ and the other is $\term{Beta}(M,1)$ which is easily seen to be the worst case. In fact for all the priors this is the worst case, by Lemma~\ref{lem:stochdomcompare} for $G,\MainLB$ and similar arguments for the others. Recall that the mean of a $\term{Beta}(a,b)$ random variable is exactly $\frac{a}{a+b}$. So, we need
\[\frac{M}{M+1+\Nexplore}<\frac{1}{M+1},\]
and so $\Nexplore=M^2.$ Since we need at most $K\cdot \Nexplore\leq KM^2$ total zero-rewards to make all arms have $0$ reward for their first $M^2$ samples, and this has probability at least $\Pexplore\geq M^{-KM^2}$, we obtain $\log(\Pexplore^{-1})=\tilde O(KM)$.

To estimate $\padG$ we show $\E[(\mu_K-\mu_q)_+]\geq M^{-O(M)}$. By Lemma~\ref{lem:sym} and the fact that $\term{Beta}(M,1)$ is the stochastically largest possible prior, we see that $\padG$ is minimized when $\mu_1,\dots,\mu_{K-1}\sim \term{Beta}(M,1)$, $\mu_K\sim \term{Beta}(1,M)$, and $q=\left(\frac{1}{K-1},\frac{1}{K-1},\dots,\frac{1}{K-1}\right).$
We will show that in this case:

\[\E[(\mu_K-\mu_q)_+]=\min(K,M)^{\Theta(M)}.\]

Because also $\E[(\mu_K-\mu_q)_+]=\frac{\Omega(1)}{\MainLB}$, this establishes both the upper and lower bounds. We first show $\E[(\mu_K-\mu_q)_+]\geq K^{-O(M)}$. Observe that with probability $K^{-O(M)}$ we have \[\mu_K>1-\frac{1}{K},\quad \mu_1,\mu_2,\mu_3<\frac{1}{3}.\] In this situation, $\mu_q\leq \frac{K-3}{K-1}<1-\frac{2}{K}$ and $\mu_K\geq 1-\frac{1}{K}$ and so $(\mu_K-\mu_q)_+\geq \frac{1}{K}$. We conclude that:

\[\E[(\mu_K-\mu_q)_+]\geq K^{-O(M)}.\]

We next show $\E[(\mu_K-\mu_q)_+]\geq M^{-O(M)}$. By Jensen's inequality we have: $\E[(\mu_K-\mu_q)_+]\geq \E[(\mu_K-\E[\mu_q])_+].$ Now, $\E[\mu_q]=1-\frac{1}{M}$ and there is at least an $M^{-O(M)}$ chance that $\mu_K\geq 1-\frac{1}{2M}$ so this shows the desired lower bound of $M^{-O(M)}.$

We now turn to the matching lower bound. Assume first that $K\leq \frac{M}{10}$. We estimate the $M/2$ exponential moment of a $\term{Beta}(1,M)$ variable denoted by $Z$. We note that $(1-x)e^x\leq 1$ for $x\in [0,1]$. Therefore
\[\mathbb E[e^{MZ/2}]=M\int_0^1 (1-x)^{M-1} e^{Mx/2}dx \leq M\int_0^1(1-x)^{M/2-1}dx= 2. \]

The random variable $1-\mu_i$ for $i\leq K-1$ is distributed as $Beta(1,M)$ and so $\E[e^{M(1-\mu_1)}]\leq 2$. Moreover since exponential moments multiply under independent sum, we obtain:
\[\mathbb E[e^{M(K-1)(1-\mu_q)/2}]\leq 2^{K-1}\leq 2^K.\]

Therefore:
\begin{align*}\mathbb P\left[\mu_q\leq 1-\frac{\log(K)}{K-1}\right]&=\mathbb P\left[1-\mu_q\geq \frac{\log(K)}{K-1}\right]\\
&\leq \mathbb E\left[e^{M(K-1)(1-\mu_q)/2}\right]\cdot e^{-\frac{M \log(K)}{2}}\\
&\leq 2^K K^{-M/2}.
\end{align*}

Since we assume $K\leq \frac{M}{10}$ we have
\[\mathbb P\left[\mu_q\leq 1-\frac{\log(K)}{K-1}\right]\leq 2^K K^{-M/2}\leq K^{-\Omega(M)}.\]

We also trivially have $\Pr[\mu_K\geq 1-\frac{\log(K)}{K-1}]\leq K^{-\Omega(M)}.$
Therefore if $K\leq \frac{M}{10}$ we have establish the upper bound $\E[(\mu_K-\mu_q)_+]\leq K^{-\Omega(M)}$. Now let us denote by $f(K,M)$ the value of $\E[(\mu_K-\mu_q)_+]$ in the worst case we are working in. Observe that $f(K,M)$ is decreasing in $K$, since decreasing $K\to (K-1)$ is equivalent to replacing $q=\left(\frac{1}{K-1},\dots,\frac{1}{K-1}\right)\in\Delta_{K-1}$ with $q'=\left(\frac{1}{K-2},\dots,\frac{1}{K-2},0\right)\in\Delta_{K-1}$, which by Lemma~\ref{lem:sym} gives a larger value of $\E[(\mu_K-\mu_q)_+]$.

To finish, it is not hard to see that the statement $f(M,K)=\min(K,M)^{-\Theta(M)}$ follows from combining the three statements below:
\begin{itemize}
    \item $ \min(K,M)^{-O(M)}\leq f(M,K)$
    \item $f(M,K)\leq K^{-\Omega(M)}$ for $K\leq \frac{M}{10}$.
    \item $f(M,K)\geq f(M,K+1)$.
\end{itemize}
Indeed, to complete the upper bound, if $K\geq \frac{M}{10}$ we have $f(M,K)\leq f(M,M/10)=M^{-\Omega(M)}$. This concludes the proof.

Finally in the case that $\mu_i,\mu_j$ have constant-separated prior means, the lower bound on $\MainLB$ follows by a simple Chernoff bound showing that
\[\mathbb E[(\mu_j-\mu_i)_+]\leq \mathbb P[\mu_j\geq \mu_i]\leq \mathbb P[\mu_j\geq t]+\mathbb P[\mu_i\leq t]\leq 2^{-\Omega(M)}
\quad \text{for $t= \nicefrac{1}{2}\rbr{\mu_i^0+\mu_j^0}$}. \qedhere\]
\end{proof}

\examplescollections*

\begin{proof}
The Gaussian case was already proved inside the proof of Corollary~\ref{thm:gaussianbound}. In the Beta case, it suffices to lower bound $\mathbb E\left[\left(\mu- 1-\frac{1}{M}\right)\right]$ for $\mu\sim \term{Beta}(1,M).$ This is at least
\[\frac{1}{2M}\mathbb P\left[\mu\geq 1-\frac{1}{2M}\right]\geq \Omega(M^{-O(M)}) \qedhere\]
\end{proof}

\realistic*

\begin{proof}








Again we assume that both Beta parameters are bounded by $M$ or $m$. It is easy to see that $\Nexplore=O(M+m)$ based on the mean value of Beta distributions, and so as before we have $\log(\Pexplore^{-1})=\tilde O(K\cdot \max(m,(1-\mu_j^0)^{-1}))$.

We next estimate $\padG$. Note that for any $q\in\Delta_K$ we have
\[\mathbb E[\mu_q]\leq \max\left(1-\frac{1}{2m},\mu_j^0\right).\]

Then we simply observe that for each $i\geq 2$ and $t\in [0,1]$, we have
\[\mathbb P[\mu_i\geq 1-t]\geq t^m.\]

In particular letting $m_0=\min\left(\frac{1}{2m},1-\mu_j^0\right)$ and taking $t=m_0/2$ we obtain:
\[\mathbb E\left[\left(\mu_i-\left(1-m_0\right)\right)_+\right]\geq \frac{m_0}{2} \mathbb P\left[\mu_i\geq \left(1-\frac{m_0}{2}\right)\right] \geq \left(\frac{m_0}{2}\right)^{m+1}.\]

This handles everything except the case of $\mathbb E[(\mu_j-\mu_q)_+]$ for $q\in \Delta_{j-1}$ (which is irrelevant when $j=1$). Note that for any $i\neq j$ we have:
\[\mathbb P\left[\mu_i\leq\frac{\mu_j^0}{2}\right]\geq \left(\frac{\mu_j^0}{2}\right)^{m}.\] Therefore we obtain $\mathbb E[(\mu_j-\mu_q)_+]\geq  \left(\frac{\mu_j^0}{2}\right)^{m(j-1)+1}$ for $q\in \Delta_{j-1}$. Altogether this implies \[\padG \geq \max(m_0,(\mu_j^0)^{(j-1)})^{O(m)}.\]

To estimate $N_{\TS}$, in the case $j=1$, for the highly informed arm $\mu_j$ we don't need to do anything since Thompson sampling is already BIC for arm $1$ at time $1$. For all the other arms $i\geq 2$, the computation above shows that
\[\mathbb P[\mu_i\geq \max_{j\neq i}\mu_j+\varepsilon_{\TS}]\geq \delta_{\TS}\]

for $\varepsilon_{\TS}\geq \frac{m_0}{2}$ and $\delta_{\TS}\geq m_0^m.$ Therefore $N_{\TS}\leq \tilde O(m_0^{-3})$. When $j\geq 2$ we also need to estimate $\varepsilon_{\TS},\delta_{\TS}$ for arm $j$. We have
\[ \mathbb E[(\mu_j-\mu_i)_+]\geq 
    \left(\mu_j^0/2\right)\mathbb P\left[\mu_i\leq\mu_j^0/2\right]\geq \left(\mu_j^0/2\right)^{m+1}.\]

Similarly the chance that arm $j$ is the best is at least

\[\delta_{\TS}\geq \Omega(1)\cdot (\mu_j^0)^{(j-1)m}\cdot 2^{-O(K)}.\]

Indeed there is a constant chance for beta random variables to be on either side of their mean, and the above therefore lower bounds the chance that $\mu_j\geq \mu_j^0$ and $\mu_i\leq \mu_j^0$ for all $i\neq j$. Applying Theorem~\ref{thm:goodalgo} or Theorem~\ref{thm:pyramidtightfast} concludes the proof.
\end{proof}

\section{Extension: A more efficient version of \MainALG}
\label{sec:improved-algo}


\SetKwFor{Loop}{loop}{}{end}

\begin{algorithm2e}[p]
\caption{\MainALG with tighter phases}
\label{alg:pyramidtightfast}

\SetAlgoLined\DontPrintSemicolon

\textbf{Parameters:}
    Desired number of samples $N$, calculated value and phase length $\padN$

\textbf{Given:} recommendation policies $\pi_1 \LDOTS \pi_K$ for \PaddedPhase.

\textbf{Initialize:} \ExplorePhase[1] of length $\max(N,\bootN,\padN)$

\For{each arm $j=2, 3  \LDOTS K$}{
    \invariant{1}{each arm $i<j$ has been sampled at least $\max(N,\bootN,\padN)$ times}

    \vspace{1mm}\tcp*[l]{Bootstrapping: two phases}
    Event $\ZEROS_{j,N_0} =
        \cbr{ \text{the first $N_0$ samples of each arm $i<j$ return reward $0$}} $. \\
    $p_j\leftarrow q/(1+q)$, where $q = \padding\cdot\Pr\sbr{\ZEROS_{j,N_0}}$.
    \\ 
    \ExploitPhase[N_0]

    \textbf{with probability} $p_j$ \textbf{do}\\
        \myTAB\ExplorePhase  \\
    \uElseIf{$\ZEROS_{j,N_0}$}
        {\PaddedPhase: use policy $\FakePi$}
                \lElse{\ExploitPhase[N]}

    \vspace{2mm}\tcp*[l]{main loop: exponentially grow the exploration probability}

    \While{$p_j< 1$}{
    \invariant{2}{$\Pr\sbr{\text{\happened} \mid \mu_1 \LDOTS \mu_K}=p_j$}

    \uIf{\happened}{
        \PaddedPhase: use policy $\pi_j$
        }
    \textbf{else with probability} $\min\rbr{1,\,\frac{p_j}{1-p_j}\cdot \padding}$ \textbf{do}\\
    \myTAB \ExplorePhase \\
    \lElse{\ExploitPhase}
Update $p_j\leftarrow \min\rbr{1,\, p_j\, (1+\padding)}$.
}

\vspace{2mm}
\tcp*[l]{Post-processing: collect remaining samples}
\For{each arm $j=1,2,\dots,K$}{
    Choose phase $\ell_0$ uniformly at random from $\sbr{ 1+\Cel{\padding^{-1}}}$ .\\
    \For{each phase $\ell=1,2,\dots,\lceil \padding^{-1}\rceil+1$}{
        \uIf{$\ell=\ell_0$}{Explore arm $j$ for $\max(N,\bootN,\padN)$ rounds}
            \uElse{Use policy $\pi_j$ for $\max(N,\bootN,\padN)$ rounds}
    }
}
} 
\end{algorithm2e}








We now give Algorithm~\ref{alg:pyramidtightfast}, a version of Algorithm~\ref{alg:pyramidtight} which requires fewer rounds. Algorithm~\ref{alg:pyramidtightfast} uses the observation that when $N,\bootN\geq \padN$ is rather large, only the initial $\padN$ samples of each arm $j$ require the hard work of exponentially growing exploration probability. As a result, we can use the same technique as in Algorithm 1 to obtain the first $N_{\post}$ samples of arm $j$, and then obtain the remaining samples more quickly. The ``post-processing" stage of each loop is easily seen to be BIC by the $(j,\padding)$-padded BIC property of policy $\pi_j$. By inspection, Algorithm~\ref{alg:pyramidtightfast} completes in this many rounds:
\begin{align}\label{eq:better-algo}
 O\rbr{ K\;\padG^{-1}\; \rbr{\padN\,\log(\padG^{-1}\;\bootP^{-1}) + \bootN + N }}.
\end{align}


\begin{theorem}\label{thm:pyramidtightfast}
Given a parameter $N\geq \padN$, Algorithm~\ref{alg:pyramidtightfast} with $N_0=\bootN$ collects $N$ rounds of each arm almost surely and completes in the number of rounds given by \eqref{eq:better-algo}.

\end{theorem}

Note that the phase length in the main part of the algorithm is only $\padN$ in Algorithm~\ref{alg:pyramidtightfast}. We remark that we do not need to assume $\bootN\leq \padN$, even though Algorithm~\ref{alg:pyramidtight} required $\bootN\leq N.$ This is because the post-processing phase ensures that Invariant 1 continues to hold, and Lemma~\ref{lem:transform} does not require $\bootN\leq N$.

\section{Extension: Improved Algorithm for ``Easy" problem Instances}
\label{sec:fastest-algo}

We now explain Algorithm~\ref{alg:pyramidsuperfast}, which achieves the guarantee of Theorem~\ref{thm:linearexplore}. We fix an $N\geq 1$ and show it samples each arm $N$ times in $\tildeO\rbr{\frac{KN}{\delta}+\frac{K}{\delta^4}}$ rounds.

The algorithm's structure is again similar to Algorithm~\ref{alg:pyramidtight}, featuring an initial bootstrap phase followed phase of exponentially growing exploration probability facilitated by a padded phase. The main difference is that we only carry out these steps for a single arm $j_0$ chosen randomly, so that we manage to sample arm $j_0$ many times without needing to first ``unlock" the previous arms. In general problem instances this may not be possible, but the $\delta$-easy assumption ensures that it is.

The algorithm continues with a \term{for} loop to complete the exploration, balanced by a padded phase. This is reminiscent of Algorithm~\ref{alg:pyramidtightfast}, but in this case we have only thoroughly explored arm $j_0$. A key new insight is that having to explore the single arm $j_0$ allows us to explore the remaining arms without requiring another exponential growth phase. This is achieved by randomizing between exploiting arm $j_0$ and exploring a random arm $i$. Because $j_0$ is random, the agent seeing the recommendation does not know whether we are exploring or exploiting, so this is BIC when the exploration probability is small. We couple the random choices of exploration arms $i$ so that there are no repeats via a uniformly random permutation $\theta:[K]\to[K]$.

\SetKwFor{Loop}{loop}{}{end}

\begin{algorithm2e}[p]
\caption{Collect Samples For Easy Collections}
\label{alg:pyramidsuperfast}
\SetAlgoLined\DontPrintSemicolon

\textbf{Parameters:}
    number of target samples $N\geq \bootN,\padN$,
    padding $\padding=\frac{\delta}{10}>0$;

\textbf{Given:} A uniformly random permutation $\theta:[K]\to[K]$

Choose a random arm $j_0$.

Event $\ZEROS_{j_0,\bootN} = \cbr{ \text{the first $\bootN$ samples of each arm $i<j_0$ return reward $0$}} $.

\tcp*[l]{Setup: get $N$ samples of arm $j_0$ with positive probability}

\For{$j=1,2,\dots,K$}{
    \ExploitPhase[\bootN] and length $\bootN$ \label{line:exploitfast1}
}
\uIf{$\ZEROS_{j_0,\bootN}$}{\ExplorePhase[j_0] with length $\padN$} \label{line:explorefast1}

\uElse{\ExploitPhase[\bootN] and length $\padN$\label{line:exploitfast1.5}}

\vspace{1mm}\tcp*[l]{Bootstrapping phase}

$p_{j_0}\leftarrow \frac{\padding\cdot\Pr\sbr{\ZEROS_{j_0,\bootN}}}{1+\padding\cdot\Pr\sbr{\ZEROS_{j_0,\bootN}}}$.

\textbf{with probability} $p_{j_0}$ \textbf{do}\\
    \myTAB\ExplorePhase[j_0] with length $\padN$ \label{line:explorefast2}\\
\uElse{

    \uIf{exploration phase has happened}{\PaddedPhase of length/depth $\padN$: use $\widehat\pi_{j_0}$ to decide whether to play arm $j$. If $\widehat\pi_{j_0}\neq j$, exploit based on all available data. \label{line:exploitfast2}}
    \uElse{\ExploitPhase[\padN] with length $\padN$.\label{line:exploitfast3}}

    } 









    \tcp*[f]{grow the exploration probability}\label{line:fastexplore}\\
    Set $p_{j_0}=\Pr\sbr{\ZEROS_{j_0,\bootN}}$.

    \While{$p_{j_0}< 1$}{
    \invariant{}{$\Pr\sbr{\text{\happened} \mid \mu_1 \LDOTS \mu_K}=p_{j_0}$}\\
    \uIf{exploration phase has happened}{
        \PaddedPhase of length and depth $\padN$: use $\widehat\pi_{j_0}$ for $\padN$ steps.
        }
    \uElse{
    With probability $\min\rbr{1,\,\frac{p_{j_0}}{1-p_{j_0}}\cdot \padding}$, \ExplorePhase[j_0] with length $\padN$; 
    \\with the remaining probability: \ExploitPhase[\padN].
}
$p_{j_0}\leftarrow \min\rbr{1,\, p_{j_0}\, (1+\padding)}$
}

\tcp*[f]{explore the other arms}\label{line:fastfinish}\\

\For{each $j=1,2,\dots,K$}{
Pick a phase $\ell_0$ uniformly at random from $\sbr{ n_{\padding} }$,
where $n_{\padding}:= 1+ \Cel{\padding^{-1}}$

\For{each phase $\ell=1,2  \LDOTS n_{\padding}$}{

\uIf{$\ell=\ell_0$}{
    Play arm $\theta(j)$ for $N$ rounds.
    }

\uElse{\PaddedPhase using $\widehat\pi_{j_0}$ with depth $\padN$ for $N$ rounds
}

}

} 

\end{algorithm2e}

\begin{lemma}\label{lem:easyexplorepadded}

Suppose $\mC$ is $\delta$-easy and $\delta$-non-dominant, and $N\geq\padN=\tilde O(\delta^{-2})$ as guaranteed by Lemma~\ref{lem:easyexploreparameters}. For fixed $j\in [K]$, consider exploitation based on $\padN$ samples from arm $j$ and no other information. This policy is $(j,\frac{\delta}{10})$-suitable.

\end{lemma}

\begin{proof}
Since we collect no information on the first $j-1$ arms, it is equivalent to replace the random values $\mu_1,\dots,\mu_{j-1}$ with their expectations $\mathbb E[\mu_i]$. From the fact that $\mC$ is $\delta$-easy to explore we see that $\padG\geq \delta$ in this case. Since we replace the $\mu_i$ with their (deterministic) expectations, we may include $\padN$ samples of each arm $a_1,\dots, a_{j-1}$ without making any difference. That a $(j,\frac{\delta}{10})$-suitable policy exists now follows from Lemma~\ref{lm:policy-existence}. However because we only observe samples of arm $j$ it is not difficult to see that for any value of $N$, exploitation based on $N$ samples from arm $j$ yields the optimal strategy in the $j$-recommendation game. The equivalence of Lemma~\ref{lem:gameBIC} implies the result.
\end{proof}

Based on the lemma above, we define for each $j\in [K]$ the policy $\widehat\pi_j$ which decides which arm to play by exploitation based on $\padN$ samples from arm $j$. In particular $\widehat\pi_j$ will only play arm $1$ or $j$.





\linearexplore*

\begin{proof}

The algorithm uses $O\left(K\bootN + \frac{\padN\log(\lambda^{-1}\Pexplore^{-1})}{\padding}+\frac{KN}{\padding}\right)=\tildeO\rbr{\frac{KN}{\delta}+\frac{K}{\delta^4}}$ rounds by construction, according to the general estimates of Lemma~\ref{lem:easyexploreparameters}. Hence we focus on the BIC property. Lines~\ref{line:exploitfast1},~\ref{line:explorefast1},\ref{line:exploitfast1.5} are clearly BIC. Line~\ref{line:explorefast1}, if executed, always samples arm $j_0$ for $\padN$ rounds. By our definition of $\widehat\pi_j$, before the second while loop recommending any arm $i\neq j_0$ is always BIC, hence we focus on the recommendations of $j_0$. We first consider the combination of Lines~\ref{line:explorefast2},~\ref{line:exploitfast2}, and~\ref{line:exploitfast3}. The proof is essentially identical to that of Lemma~\ref{lm:main-BIC0}, where the point is that conditioned on reaching line~\ref{line:exploitfast2}, playing from $\widehat\pi_{j_0}$ is $(j,\padding)$-padded BIC. That we condition on $\ZEROS_{j_0,\bootN}$ only helps the $(j,\padding)$-padded property, because this conditioning decreases the mean reward of arm $i$ for all $i<j_0$. This counterbalances the exploration in line~\ref{line:explorefast2}.

The \texttt{while} loop is BIC for the same reason as in the proof of \ref{lm:main-BIC}. The key point is again that the padded, exploration, and exploitation phases occur independently of the true mean rewards. To show that the final \texttt{for} loop is BIC,
we observe:
\[\mathbb E[(\mu_j-\mu_i)\cdot 1_{A_t=a_j}] \geq \frac{1}{K}\cdot\left( \padding\cdot \frac{n_{\padding}-1}{n_{\padding}}-\frac{1}{n_{\padding}}\right)\geq 0.\]
Here the first term comes from the exploitation phase while the second term comes from the event $t=s$. The factor $\frac{1}{K}$ comes from the randomness in choosing $j_0$ and $\theta:[K]\to[K].$ This concludes the proof that the algorithm is BIC.
\end{proof}


\section{Extension: Efficient Computation for Beta Priors}
\label{sec:computation-beta}

\betaefficient*



\begin{proof}

First suppose that we are in the worst case $\mu_1,\dots,\mu_{j-1}\sim \term{Beta}(M,1), \mu_j\sim  \term{Beta}(1,M)$ and aim to explore $a_j$. The key point is that by the argument of Lemma~\ref{lem:sym}, assuming without loss of generality that our strategy is symmetric in $(\mu_1,\dots,\mu_{j-1})$, the uniform distribution $q=q_j:=\left(\frac{1}{j-1},\frac{1}{j-1},\dots,\frac{1}{j-1}\right)$ over arms $i<j$ is always a best response in the $j$-recommendation game. Therefore the optimal $j$-recommendation strategy is to recommend arm $K$ exactly when $\widetilde\mu_j\geq \widetilde\mu_{q_j}.$ Here as usual we use $\widetilde\mu_i$ to denote posterior mean with respect to the relevant data, in this case the first $N$ samples of arm $j$.

To efficiently compute the resulting value of $\padding=\mathbb E[(\widetilde\mu_j\geq \widetilde\mu_{q_j})_+]/10$ is not difficult. One can simply compute the distribution of $\mathbb E[\mu_{i}|\mathcal G]$ for each $i<j$ and then compute convolutions to find the distribution of $\mathbb E[\mu_{q_0}|\mathcal G].$ Since Beta distributions with a fixed strength have closed-form probability mass functions supported on an arithmetic progression, this is computationally efficient.

Of course, we might not have $\mu_1,\dots,\mu_{j-1}\sim \term{Beta}(M,1)$, $\mu_K\sim  \term{Beta}(1,M)$. The second key point is that we may reduce to this case in a similar manner to the construction of $\FakePi$ previously. Indeed, let $\mu_i'\sim\term{Beta}(M,1)$. Likewise let $\mu_j'\sim \term{Beta}(1,M)$, and let $\mathcal G'=\mathcal G'_{N,j}$ be a $\sigma$-algebra encapsulating $N$ samples of arms with means $\mu_1',\dots,\mu_{j}'.$

By Corollary~\ref{cor:betapostdom} we know that $\widetilde\mu_i=\mathbb E[\mu_i|\mathcal G]$ is stochastically smaller than $\mathbb E[\mu_i'|\mathcal G']$, for each $i<j$, and that the opposite holds for $\mu_j,\mu_j'$. Moreover it is computationally easy to compute these distributions exactly.\footnote{For instance, the sequence of $0/1$-reward values with a $\term{Beta}$-prior is a simple Markov chain following Laplace's rule of succession (with prior-dependent initialization), so the probabilities to earn $k$ reward from $N$ samples can be computed easily by dynamic programming.} Once the two distributions are computed we then compute the canonical monotone coupling between the two conditional expected values, as in defining the couplings $\mathcal X_i$ in the description of Algorithm~\ref{alg:pyramidtight}. Finally, given the values $\widetilde\mu_i=\mathbb E[\mu_i|\mathcal G]$ we sample the value $\hat\mu_i=\mathbb E[\mu_i'|\mathcal G']$ according to $\mathcal X_i$, for each $i\leq j$. Hence we have
    $\widetilde\mu_i\leq \hat\mu_i$
if $i<j$, and
    $\widetilde\mu_j\geq \hat\mu_i$ otherwise.

$\pi_j^{\term{eff}}$ can now be defined. $\pi_j^{\term{eff}}$ first decides whether to recommend arm $j$, doing so whenever
\[\hat\mu_j\geq\hat\mu_{q}=\frac{1}{j-1}\sum_{i\in [j-1]}\hat\mu_i.\]
When this does not happen, $\pi_j^{\term{eff}}$ exploits conditional on $\mathcal G$ as usual.

To see that $\pi_j^{\term{eff}}$ is $(j,\padding)$-padded BIC we use the same strategy as in previous comparison arguments. Exploitation parts are automatically BIC. For $i<j$ we have:
\begin{align*}
\E[(\mu_j-\mu_i)\cdot 1_{\pi_j^{\term{eff}}=j}] &=\E[(\widetilde \mu_j-\widetilde \mu_i)\cdot 1_{\pi_j^{\term{eff}}=j}]\\
&\geq \E[(\hat \mu_j-\hat \mu_i)\cdot 1_{\pi_j^{\term{eff}}=j}]\\
& \geq \padding.\end{align*}

The ordinary BIC property against the other arms $i>j$ holds similarly.
\end{proof}

\end{document}